\documentclass[12pt]{article}

\usepackage{amsmath,amsthm,amsfonts,amssymb}
\usepackage{bm}

\usepackage{graphicx}
\usepackage{caption}
\usepackage{subcaption}
\usepackage{float}
\usepackage{epstopdf}
\AppendGraphicsExtensions{.tiff}

\usepackage{booktabs}
\usepackage{longtable}

\usepackage[margin=1in]{geometry}
\usepackage[nodisplayskipstretch]{setspace}
\usepackage{xcolor}
\usepackage{url}
\usepackage{enumerate}

\usepackage{algorithm,algpseudocode}

\newcounter{phase}[algorithm]
\newlength{\phaserulewidth}

\newcommand{\phase}[1]{%
  \vspace{-1.25ex}
  \Statex\leavevmode\llap{\rule{\dimexpr\labelwidth+\labelsep}{\phaserulewidth}}\rule{\linewidth}{\phaserulewidth}
  \Statex\strut\refstepcounter{phase}\textit{Phase~\thephase~--~#1}%
  \vspace{-1.25ex}\Statex\leavevmode\llap{\rule{\dimexpr\labelwidth+\labelsep}{\phaserulewidth}}\rule{\linewidth}{\phaserulewidth}}

\usepackage{natbib}

\usepackage[colorlinks=true,linkcolor=blue,citecolor=blue,urlcolor=blue]{hyperref}

\newtheorem{theorem}{Theorem}[]

\newtheorem{lemma}[theorem]{Lemma}
\theoremstyle{definition}
\newtheorem{definition}{Definition}[section]

\newcommand{\beginsupplement}{%
    \setcounter{section}{0}%
    \setcounter{subsection}{0}%
    \setcounter{equation}{0}%
    \setcounter{figure}{0}%
    \setcounter{table}{0}%
    \renewcommand{\thesection}{S\arabic{section}}%
    \renewcommand{\thesubsection}{S\arabic{section}.\arabic{subsection}}%
    \renewcommand{\theequation}{S\arabic{equation}}%
    \renewcommand{\thefigure}{S\arabic{figure}}%
    \renewcommand{\thetable}{S\arabic{table}}%
    \renewcommand{\theHsection}{Supplement.\arabic{section}}%
    \renewcommand{\theHsubsection}{Supplement.\arabic{section}.\arabic{subsection}}%
    \renewcommand{\theHequation}{Supplement.\arabic{equation}}%
    \renewcommand{\theHfigure}{Supplement.\arabic{figure}}%
    \renewcommand{\theHtable}{Supplement.\arabic{table}}%
}

\newcommand{\bpm}{\begin{pmatrix}}
\newcommand{\epm}{\end{pmatrix}}

\allowdisplaybreaks

\title{Regression and Dimension Reduction for Multivariate Mixed-Type Data\\
via Semiparametric Gaussian Copula}

\author{
  Debangan Dey\thanks{Corresponding author. \texttt{debangan@tamu.edu}}\\
  \small Department of Statistics, Texas A\&M University, College Station, TX, USA
  \and
  Vadim Zipunnikov\\
  \small Department of Biostatistics, Johns Hopkins Bloomberg School of Public Health, Baltimore, MD, USA
}

\date{}

\begin{document}

\maketitle

\begin{abstract}
Clinical and epidemiological studies encode participant information in multivariate vectors with mixed type variables on continuous, truncated, ordinal, and binary scales. Semiparametric Gaussian Copula (SGC) assumes that observed data is generated by latent multivariate normal random variables which marginals are monotonically transformed and then truncated/ordinalized/binarized. In SGC, the latent correlation matrix fully determines the dependence structure and it is 
estimated through an inversion of ``bridges'' between Kendall’s Tau rank correlations of observed variables and latent correlations. By employing SGC, we develop regression (SGC-Reg), principal component analysis (SGC-PCA), and principal component regression (SGC-PCR) for latent representations of observed data. To build our framework, we make several key contributions: i) establishing novel bridging results for general ordinal type variables, ii) developing regression estimation on the latent space and deriving asymptotic normality of estimators, iii) developing a computationally efficient algorithm that reduces calculation complexity of all steps including calculation of asymptotic covariance matrix from $O(n^4)$ to $O(n\log n)$, iv) developing methods to predict latent representations of observed data and perform imputation of missing data, and v) developing principal component analysis and principal component regression on the latent space. We apply our framework to study the association between a 5-year mortality and 61 frailty-related measures composed of 29 continuous, 17 ordinal, and 15 binary variables in 9478 participants of 1999-2010 waves of National Health and Nutrition Examination Survey (NHANES).
\end{abstract}


\section{Introduction}
\label{sec:intro}

To gain insights into human health, national health surveys and biobanks collect multivariate records with variables measured on different scales, including binary, ordinal, truncated, and continuous. Such mixed-type data are common in studies of aging and frailty, where adverse health outcomes reflect decline across multiple physiological systems \citep{fried2001frailty, xue2011frailty}. For example, the National Health and Nutrition Examination Survey (NHANES) is frequently used to study frailty among older adults \citep{blodgett2015frailty}. Motivated by such applications, we develop a statistical framework for estimating joint dependence in multivariate mixed-type data and incorporating that dependence into regression and dimension reduction. In the remainder of this section, we review existing approaches for conditional and joint modeling of mixed data, then describe the proposed framework and summarize our main contributions.

Due to the lack of standard joint models for multivariate mixed data, conditional modeling is often used instead. These approaches select one variable as the outcome and model its conditional mean as a function of the others: logistic and probit regressions are common for binary outcomes, cumulative ordinal regression is widely used for ordinal outcomes \citep{mccullagh1980regression}, and Tobit models for truncated outcomes assume truncation of a latent normal variable whose conditional mean depends linearly on predictors \citep{tobin1958estimation, heckman1976common, hausman1977social}. The Hurdle model extends Tobit by separately modeling truncated and non-truncated components, at the cost of additional parameters \citep{cragg1971some}. These methods are typically likelihood-based and can suffer from convergence issues. Factorization models provide an alternative by combining conditional and marginal distributions. Examples include General Location Models (GLOMs) \citep{olkin1961multivariate}, which impose conditional normality for continuous variables and arbitrary distributions for discrete components, and conditional grouped conditional models (CGCMs) \citep{anderson1985grouped, de2007general, de2005pairwise}, which treat discrete variables as truncations of latent continuous variables and use polychoric and polyserial correlations to estimate a joint covariance structure. Although convenient for specifying mixed distributions, factorization models induce a hierarchy that depends on the conditioning order, so different factorizations of the same variables can lead to different parameter interpretations and inferences.

Another class is copula-based models. For example, \citep{song2009joint} defined Vector Generalized Linear Models (VGLMs) for modeling mixed data type outcomes. This approach requires embedding the marginal distributions (univariate Generalized Linear Models) into the joint distribution function via a Gaussian copula. Gaussian copulas is a popular choice to couple marginal distributions because of their analytical tractability and flexibility. \cite{jiryaie2016gaussian} introduced Gaussian copula distributions (GCD) that take a latent variable approach to embed discrete variables using the Gaussian copula. However, CGCMs, GLOMs, VGLMs, and GCDs require likelihood-based inference and are computationally intensive for high dimensions. Pairwise likelihood-based approaches \citep{de2005pairwise, jiryaie2016gaussian} reduce the computational burden but can make worse classification than the full likelihood-based approach  \citep{jiryaie2016gaussian}. 

Recently, semiparametric models have seen a wide adaptation for joint modeling of multivariate mixed type data. \cite{wang2014semiparametric} used likelihood-based inference and \cite{cai2015high} developed a rank-based approach to estimate a joint semiparametric Gaussian copula for continuous variables. \cite{fan2016multitask} extended the rank-based approaches to perform quantile regression on continuous variables. Rank-based estimation of the covariance of the semi-parametric Gaussian copula family \citep{liu2009nonparanormal, liu2012high} has been particularly attractive because of the fast and robust estimation procedure. In addition, multiple recent extensions used latent semi-parametric Gaussian copula to model mixed types data. \cite{fan2017high} developed the estimation for the case of  binary and continuous variables. \cite{yoon2020sparse, wang2025truncated} extended the approach to include truncated variables, and \cite{quan2018rank} has additionally extended it to include ternary variables (ordinal variables with three categories) and general ordinal-continuous pairs of variables. \cite{feng2019high} represented an ordinal variable via multiple dummy binary variables and took a weighted correlation approach to recover the latent correlation for ordinal pairs with more than three categories. Whereas, \cite{zhang2018high} arrived at an incorrect bridging function trying to tackle the general ordinal case. \cite{huang2021latentcor, yoon2021fast} provided an R package for speeding up numerical calculation of latent correlations between pairs of binary, ternary, truncated and continuous variables using a numerical interpolation approach. 

Leveraging the semiparametric Gaussian copula, we develop a unified framework for estimating the joint dependence structure of multivariate mixed data with binary, ordinal, truncated, and continuous variables. Building on this latent dependence model, we introduce regression (SGC-Reg), principal component analysis (SGC-PCA), and principal component regression (SGC-PCR) for latent representations of the observed data. Our main contributions are: i) new bridging results for general ordinal variables, ii) latent-space regression with asymptotic normality of the estimators, iii) a computationally efficient algorithm that reduces the complexity of all major steps, including asymptotic covariance estimation, from $O(n^4)$ to $O(n\log n)$, iv) methods for latent prediction and missing-data imputation, and v) latent-space PCA and PCR. The resulting framework is semiparametric, likelihood-free, and computationally efficient, while providing mutually consistent conditional modeling across mixed outcomes, automatic normalization across heterogeneous scales, invariance to monotone transformations, interpretable latent-scale effect sizes, and a principled way to address multicollinearity through latent PCR.

The rest of the paper is organized as follows. Section 2 reviews the semiparametric Gaussian copula model and develops handling of general ordinal variables. Section 3 introduces SGC-based regression, principal component analysis, and principal component regression, and derives the corresponding asymptotic results. Section 4 presents two methodological applications: latent variable prediction and missing-data imputation. Section 5 evaluates the proposed methods through simulation studies. Section 6 applies the framework to frailty-related variables and mortality in NHANES. Section 7 concludes with a discussion.




\section{Gaussian copula model}

Classical Gaussian model assumptions have been popular due to their computational simplicity. However, these assumptions can be too restrictive. As an alternative, \cite{liu2009nonparanormal} proposed non-Paranormal distribution (NPN) which can be seen as a semiparametric Gaussian Copula model. 

\theoremstyle{definition}
\begin{definition}{(Non-paranormal distribution)} A random vector $Z=(Z_1,\dots,Z_p)' \sim NPN_p(0,\Sigma,f)$, if there exist monotone transformation functions $f=(f_1,...,f_p)$ such that $L = (L_1, \dots, L_p) = f(Z)=(f_1(Z_1),\dots, f_p(Z_p)) \sim N(0,\Sigma)$, where $\Sigma_{jj}=1$ for $1 \le j \le p$. 
\end{definition}

The assumption on $\Sigma$ is made to ensure the identifiability of the distribution as shown in \cite{liu2009nonparanormal}.

\cite{fan2017high} introduced latent non-paranormal distribution which extended non-paranormal distribution to jointly model binary and continuous data. \cite{yoon2020sparse} introduced truncated variables using latent non-paranormal variables and \cite{quan2018rank} extended the distribution to ternary-continuous pairs and ternary-ternary pairs. We generalize these results even further and treat a (general) ordinal ($k$ ordered categories) case. We define Generalized Latent Non-paranormal (GLNPN) distribution with four mixed data types including continuous, truncated, (general) ordinal ($k$ ordered categories), binary variables. 

\theoremstyle{definition}
\begin{definition}{(Generalized latent non-paranormal distribution)} Suppose we observe a random vector $X=(X_c, X_t, X_o, X_b)'$, where  $X_c$ is $p_c$-dimensional continuous variable, $X_t$ is $p_t$-dimensional truncated, $X_o$ is $p_o$-dimensional ordinal ($j$-th ordinal variable has levels $\{0,1,\cdots,l_j-1\}$), and $X_b$ is $p_b$-dimensional binary variable, and $p = p_c+p_t+p_o+p_b$. We assume that there exist latent variables $Z=(Z_c,Z_t,Z_o,Z_b)'$ such that 
\begin{equation}
\begin{aligned}
 X_{cj} &= Z_{cj}, 1 \le j \le p_c\\
 X_{tj} &= Z_{tj} I(Z_{tj} > \delta_{tj}), 1 \le j \le p_t \\
 X_{oj} &= \sum_{r=0}^{l_j-1} r I(\delta_{ojr} \le Z_{oj} < \delta_{oj(r+1)}), 1 \le j \le p_o; \delta_{oj0}= -\infty, \delta_{ojl_j}=\infty \\
 X_{bj} &= I(Z_{bj} > \delta_{bj}), 1 \le j \le p_b 
\end{aligned}
\end{equation}
If $Z=(Z_c, Z_t, Z_o, Z_b)' \sim NPN(0,\Sigma,f)$, we denote that $X=(X_c, X_t, X_o, X_b)'\sim \mathrm{GLNPN}_p(0,\Sigma,f,\boldsymbol{\delta})$, where $\boldsymbol{\delta} = \{\delta_{tj}; j = 1, \cdots, p_t\} \cup (\cup_{j=1}^{p_o}\{\delta_{oj(r+1)}; r = 0, \cdots, l_j\})\cup \{\delta_{bj}; j=1, \cdots, p_b\}$, i.e., is the set containing cutoffs for truncated, ordinal and binary variables.
\end{definition}
\begin{figure}[htp]
    \centering
    \includegraphics[scale=0.5, trim=175 175 175 175]{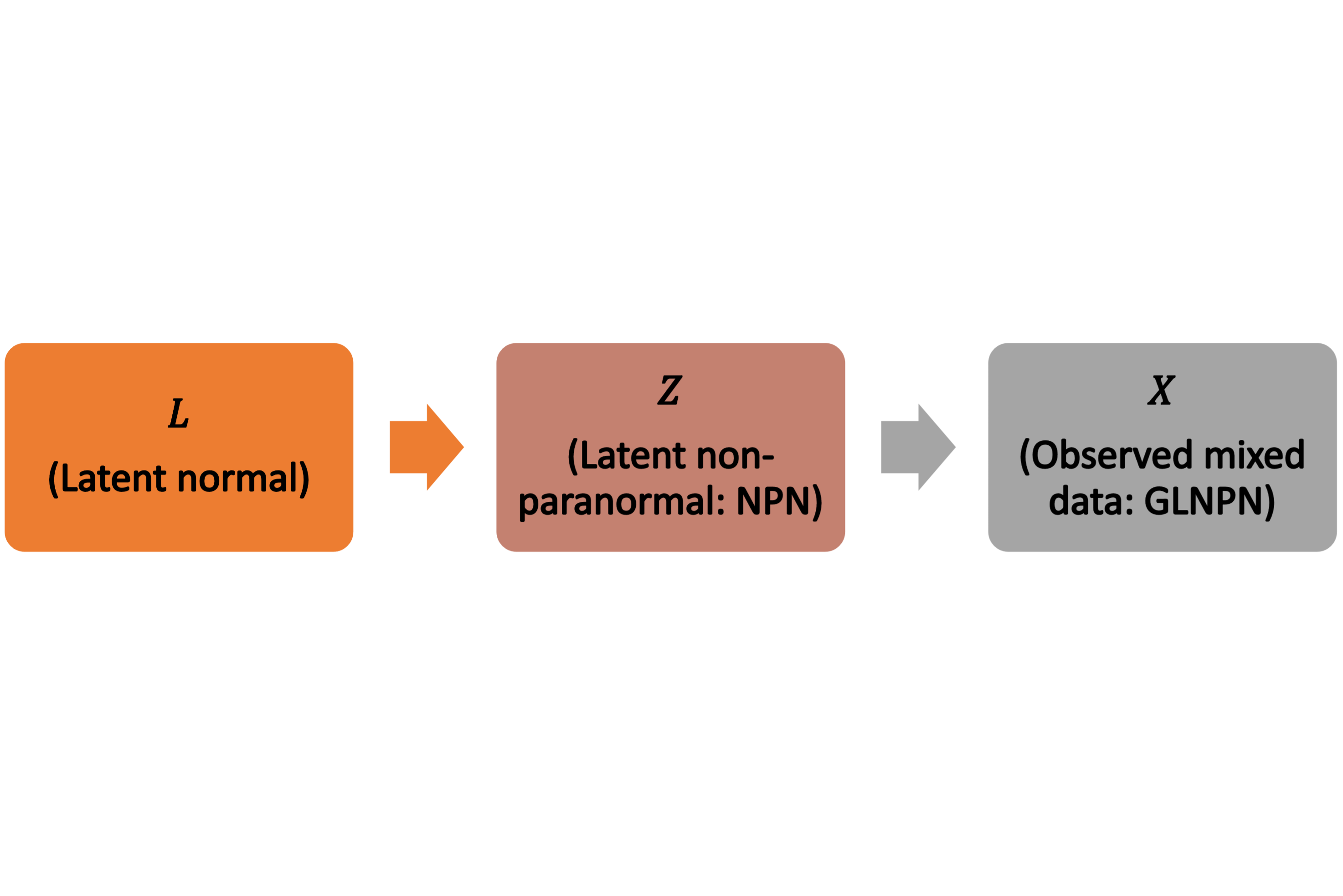}
    \caption{The data generation flow of GLNPN distribution}
    \label{fig:dataflow-glnpn}
\end{figure}

For notational simplicity, we will refer to observed continuous, truncated, ordinal, binary variables generated according to GLNPN distribution as CTOB-GLNPN variables. Figure \ref{fig:dataflow-glnpn} shows the flowchart of the data generation mechanism for observed GLNPN variables. Figure \ref{fig:illust-glnpn} shows an example of four observed CTOB-GLNPN variables generated via  monotone-transformation-then-truncation of latent bivariate normal variables.

\begin{figure}[htp]
    \centering
    \includegraphics[scale=0.105]{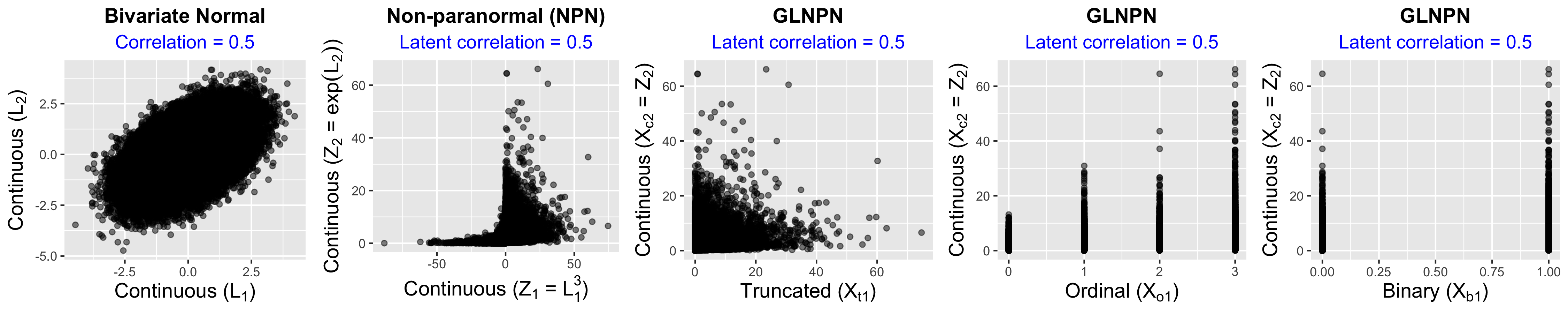}
    \caption{From left to right: (i) a scatterplot of  bivariate standard normal variables with correlation of $0.5$,  (ii) a continuous-continuous pair, (iii) a truncated-continuous pair, iv) an ordinal-continuous pair, (v) a binary-continuous pair}
    \label{fig:illust-glnpn}
\end{figure}

Cut-off parameters of GLNPN distribution suffers from identifiability issues  (discussed in details at \cite{fan2017high}). The joint probability mass function of the discrete component or the density of the truncated component only depends on the set of transformed cutoffs: $\mathbf{\Delta} = f(\boldsymbol{\delta}) = \{f_j(\delta_{tj}); j = 1, \cdots, p_t\} \cup (\cup_{j=1}^{p_o}\{f_j(\delta_{oj(k+1)}); k = 0, \cdots, l_j\})\\ \cup \{f_j(\delta_{bj}); j=1, \cdots, p_b\}= \{\Delta_{tj}; j = 1, \cdots, p_t\} \cup (\cup_{j=1}^{p_o}\{\Delta_{oj(k+1)}; k = 0, \cdots, l_j\})\\ \cup \{\Delta_{bj}; j=1, \cdots, p_b\}$. To emphasize this, we will generally refer to GLNPN distribution as $\mathrm{GLNPN}_p(0,\Sigma,f,\mathbf{\Delta})$. 

As a result of the identifiability constraints for the cutoffs, the binary and ordinal components of GLNPN distribution are marginally equivalent to the latent Gaussian distribution for binary and ordinal variables. This comes as no surprise as the discrete components does not have enough information to identify the marginal transformations. However, when we model the discrete component jointly with continuous and truncated variables, the class of GLNPN distributions becomes much larger than the class of latent Gaussian distributions \citep{yoon2020sparse, fan2017high}. The marginal transformations from continuous and truncated variables make the joint distribution of mixed variables more flexible and potentially can provide a substantial advantage to better explain the association between mixed type of variables. 

\subsection{Estimation of Correlation Matrix}\label{sec: estimation}

\subsubsection{Bridging functions}
A few authors (including \cite{fan2017high, yoon2020sparse, quan2018rank}) have considered Kendall's $\tau$ rank correlation to estimate latent correlation matrix $\Sigma$ for different combinations of continuous, truncated, ternary, and binary variables. For two independent copies $X_i$, $X_i'$ of the random vector $X$, the population-level Kendall's $\tau$ is defined as  
\begin{equation}
\tau_{jk} = E[sgn\{(X_{ij}-X_{i'j})(X_{ik} - X_{i'k})\}]
\end{equation}

We can calculate a sample Kendall's tau between $j$-th and $k$-th variable as follows:
\begin{equation}
\hat{\tau}_{jk} = \frac{2}{n(n-1)} \sum_{1\le i < i' < n} sgn\{(X_{ij}-X_{i'j})(X_{ik} - X_{i'k})\}
\end{equation}

The construction of Kendall's $\tau$ reveals that it is invariant under a monotone transformation. 

Under GLNPN, the population-level Kendall's Tau can be related to the latent correlation $\Sigma_{jk}$ through a one-to-one bridging function $F$ which for non-continuous components will depend on cutoffs as $\tau_{jk}=F(\Sigma_{jk})$. The estimated latent correlation is then obtained as $\hat{\Sigma}_{jk} = F^{-1}(\hat{\tau}_{jk})$.

\cite{fan2017high} calculated the bridging function for a pair of binary and continuous variables, \cite{yoon2020sparse} showed how to deal with truncated variables in addition to continuous and binary. \cite{quan2018rank} provided formulas for bridging functions for ternary variables and for general ordinal-continuous pairs of variables. \cite{feng2019high} broke an ordinal variable into multiple dummy binary variables and took a weighted correlation approach to recover the latent correlation for ordinal pairs with more than three categories, whereas, \cite{zhang2018high} arrived at an incorrect bridging function trying to tackle the general ordinal case. We summarize the references to the correct bridging functions for all possible pairs of variables in Table \ref{tab:ref-bridge}. Our first contribution is to derive bridging functions for the general ordinal variable with an arbitrary number of ordinal levels. We derive bridging functions for ordinal-truncated, ordinal-ordinal, and ordinal-binary pairs in Theorem $\ref{thm:bridge}$ and prove the invertibility of the functions in Theorem \ref{thm:bridge-inv}. 

\begin{table}[H]
\begin{tabular}{|l|llll|}
\hline
Type                    & Continuous & Truncated & Ordinal$^{*}$ & Binary \\ \cline{2-5} 
\hline
\multicolumn{1}{|l|}{Continuous} & \cite{liu2009nonparanormal}  & \cite{yoon2020sparse}  & \cite{quan2018rank}   & \cite{fan2017high}  \\
\multicolumn{1}{|l|}{Truncated} & \cite{yoon2020sparse}  & \cite{yoon2020sparse}  & Theorem \ref{thm:bridge} & \cite{yoon2020sparse}  \\
\multicolumn{1}{|l|}{Ordinal} & \cite{quan2018rank}   & Theorem \ref{thm:bridge}  & Theorem \ref{thm:bridge}  &  Theorem \ref{thm:bridge} \\
\multicolumn{1}{|l|}{Binary} & \cite{fan2017high}  & \cite{yoon2020sparse}  &  Theorem \ref{thm:bridge} & \cite{fan2017high}  \\ \hline
\end{tabular}
\caption{The reference of bridging functions for all possible pairs of variables.\\ \footnotesize{\emph{$^{*}$Ordinal cases for only three categories were derived in \cite{quan2018rank}}}}
\label{tab:ref-bridge}
\end{table}

For completeness, the following theorem also includes the previously established bridging functions.  

\begin{theorem}
\label{thm:bridge}
Let $X_j, X_k$ be two GLNPN variables, then the population Kendall's Tau is related to the latent correlation as follows: $\tau_{jk} = F(\Sigma_{jk})$, where $F$ additionally depends on the cutoffs ${\Delta_j}, {\Delta_k}$ for non-continuous components. The bridging functions corresponding to all pairs of variables are as follows

\begin{align}
     F_{\rm cc}(\rho) & = \dfrac{2}{\pi} \sin^{-1}(\rho) \nonumber\\
     F_{\rm bb}(\rho; \Delta_j, \Delta_k)&  = 2\left\{\Phi_2( \Delta_j,  \Delta_{k}; \rho)-\Phi( \Delta_j)\Phi( \Delta_{k})\right\} \nonumber\\
     F_{\rm cb}(\rho; \Delta_j) & = 4\Phi_2( \Delta_j, 0; \rho/\sqrt{2})-2\Phi( \Delta_j) \nonumber\\
     F_{\rm tb}(\rho; \Delta_j, \Delta_k) & =
    2 \Phi_3\left(\Delta_k, -\Delta_j, 0; S_{3a}(\rho) \right)
    -2 \Phi_3\left(\Delta_k, -\Delta_j, 0; S_{3b}(\rho) \right) \nonumber\\
    F_{\rm ct}(\rho; \Delta_j) & = -2 \Phi_2 (-\Delta_j,0; 1/\sqrt{2} ) +4\Phi_3 \left(-\Delta_j,0,0; S_3(\rho)\right) \nonumber\\
    F_{\rm tt}(\rho; \Delta_j, \Delta_k)  & = -2 \Phi_4 (-\Delta_j, -\Delta_k, 0,0; S_{4a}(\rho)) + 2 \Phi_4 (-\Delta_j, -\Delta_k, 0,0; S_{4b}(\rho)) \nonumber\\
    F_{\rm co}(\rho; \Delta_j) & = \sum_{r=1}^{l_j-1} (4 \Phi_3 (\Delta_{jr}, \Delta_{j(r+1)},0; S_{3c}(\rho)) - 2 \Phi(\Delta_{jr})\Phi(\Delta_{j(r+1)})) \nonumber\\
    F_{\rm oo}(\rho; \Delta_j, \Delta_k) &= 2 \sum_{r=1}^{l_j-1}\sum_{s=1}^{l_k-1}[\Phi_2(\Delta_{jr},\Delta_{ks}; \rho)\{ \Phi_2(\Delta_{j(r+1)},\Delta_{k(s+1)};\rho) - \Phi_2(\Delta_{j(r+1)},\Delta_{k(s-1)}; \rho)\}] - \nonumber \\
& 2 \sum_{r=1}^{l_j-1}\Phi({\Delta_{jr}})\Phi_2(\Delta_{j(r+1)},\Delta_{k(l_k-1)}; \rho) \nonumber \\
 F_{\rm ob}(\rho; \Delta_j, \Delta_k) &= 2 \sum_{r=1}^{l_j-1}[\Phi_2(\Delta_{jr},\Delta_{k}; \rho)\Phi(\Delta_{j(r+1)}) - \Phi({\Delta_{jr}})\Phi_2(\Delta_{j(r+1)},\Delta_{k}; \rho)] \nonumber \\ 
F_{\rm to}(\rho; \Delta_j, \Delta_k) & =2 \Phi_3( (\Delta_{k(l_k-1)}, -\Delta_j, 0); S_{3a}(\rho)) - 2 \sum_{r=1}^{l_{j}-1}[\Phi_4((\Delta_{k(r+1)}, \Delta_{kr}, -\Delta_j, 0); S_{5}(\rho))- \nonumber \\
& \Phi_4((\Delta_{k(r-1)}, \Delta_{kr}, -\Delta_j, 0); S_{5}(\rho))]\label{eqn:bridge-thm1}
\end{align}
with 
\begin{equation}
    \begin{split}\nonumber
        S_{3a}(\rho)& =\begin{pmatrix} 1 & 0 & \frac{\rho}{\sqrt{2}} \\
0 & 1 & \frac{1}{\sqrt{2}} \\
\frac{\rho}{\sqrt{2}}  & \frac{1}{\sqrt{2}} & 1
\end{pmatrix}, \qquad
S_{3b}(\rho)=\bpm
1 & 0 & -1/\sqrt{2} \\
0 & 1 & -\rho/\sqrt{2}\\
-1/\sqrt{2} & -\rho/\sqrt{2} & 1
\epm \\
S_{3c}(\rho) &= \begin{pmatrix}
    1 & 0 & \rho/\sqrt{2}\\
    0 & 1 & -\rho/\sqrt{2}\\
    \rho/\sqrt{2} & -\rho/\sqrt{2} & 1
\end{pmatrix}, \qquad
S_3(\rho) = \bpm
        1 & -\rho & -\rho/\sqrt{2}\\
        -\rho & 1 & \frac{1}{\sqrt{2}}\\
        -\rho/\sqrt{2} & \frac{1}{\sqrt{2}} & 1
    \epm \\
S_{4a}(\rho) &= \bpm
1 & 0 & 1/\sqrt{2} & -\rho/\sqrt{2}\\
0& 1 & -\rho/\sqrt{2}& 1/\sqrt{2}\\
1/\sqrt{2}& -\rho/\sqrt{2} & 1& -\rho\\
-\rho/\sqrt{2}& 1/\sqrt{2}& -\rho & 1
\epm\\
S_{4b}(\rho) & = \bpm
    1 & \rho & 1/\sqrt{2} & \rho/\sqrt{2}\\
    \rho & 1 & \rho/\sqrt{2}& 1/\sqrt{2}\\
    1/\sqrt{2}& \rho/\sqrt{2} & 1& \rho \\ \rho/\sqrt{2}& 1/\sqrt{2}& \rho  & 1
\epm, \quad
    S_{5}(\rho)  = \begin{pmatrix} 1 & 0 &  0 & \frac{\rho}{\sqrt{2}} \\
0 & 1 & -\rho & -\frac{\rho}{\sqrt{2}} \\
0  & -\rho & 1 & \frac{1}{\sqrt{2}} \\
\frac{\rho}{\sqrt{2}} & -\frac{\rho}{\sqrt{2}} & \frac{1}{\sqrt{2}} & 1
\end{pmatrix}
\end{split}
\end{equation}
where $\Phi$ denotes the CDF of a univariate standard normal random variable, $\Phi_d(\dots,S)$ denotes the CDF of a $d$-variate standard normal random variable with correlation matrix $S$. For notational simplicity, we use ${\Phi}_2(\dots,\rho)$ as the standard bivariate normal cdf with correlation $\rho$.
\end{theorem}

\begin{proof}
The derivations of $F_{cc}, F_{bb}, F_{cb}, F_{tb}, F_{ct}, F_{tt}, F_{co}$ have been previously shown in literature as reported Table \ref{tab:ref-bridge}. Novel derivations of $F_{oo}, F_{ob}, F_{to}$ are provided in  \ref{sec: proofs}. To the best of our knowledge, this theorem is the first result deriving analytical forms of pairwise bridging functions for ordinal-truncated, ordinal-ordinal,  ordinal-binary pairs when the ordinal variable has an arbitrary number of levels.
\end{proof}

\begin{theorem}
\label{thm:bridge-inv}
For constants $\Delta_j, \Delta_k$, the bridging functions $F(\rho)$ in Theorem $\ref{thm:bridge}$ are strictly increasing in $\rho \in (-1,1).$ Hence, the inverse, $F^{-1}(\tau_{jk})$ exists.
\end{theorem}

\begin{proof}
We present the proofs of the invertibility of  $F_{ob}, F_{oo}, F_{to}$ in \ref{sec: proofs}. The invertibility of the other bridging functions have been shown in \cite{fan2017high, quan2018rank, yoon2020sparse}. 
\end{proof}

Theorem \ref{thm:bridge} shows that the bridging function depends on the cutoffs for pairs involving binary, ordinal and truncated variables. From the observed data, we can estimate the cutoffs through the method of moments as follows: 

\begin{equation}
\begin{split}
\textrm{Binary: } \widehat{\Delta}_j & =\Phi^{-1}\left(\frac{\sum_{i=1}^{n}I(X_{ij}=0)}{n}\right)\\
\textrm{Ordinal: } \widehat{\Delta}_{jr} & = \Phi^{-1}\left(\frac{\sum_{i=1}^{n}I(X_{ij}<=(r-1))}{n}\right), r=1, \dots, l_{j}-1\\
\textrm{Truncated: } \widehat{\Delta}_j & =\Phi^{-1}\left(\frac{\sum_{i=1}^{n}I(X_{ij}=0)}{n}\right)
\end{split}
 \label{eqn:cutoff-est}
\end{equation}

Next, we can plug-in estimated cutoffs in  bridging functions from \eqref{eqn:bridge-thm1}, so bridging functions now only depend on latent correlations. After bridging, the correlation matrix is formed as $\hat{\Sigma} = (\hat{\Sigma}_{jk})$. Note that the estimated matrix is not guaranteed to be positive semi-definite. So, we need to perform an extra step and find the nearest positive-definite correlation matrix \citep{higham2002computing}. In Algorithm \ref{alg: estimation}, we put together all steps of our estimation procedure.

\begin{singlespace}
\begin{algorithm}[h]
\caption{GLNPN estimation algorithm}\label{alg: estimation}
\begin{algorithmic}[1]
\State Input: Observed data, $X_i= (X_{ic}, X_{it}, X_{io}, X_{ib}), i=1, \dots, n$
\phase{Estimating cutoffs}
\For{$j$ in $\{t,o,b\}$}
    \State Estimate the set of cutoffs $\widehat{\Delta}_j$ from \eqref{eqn:cutoff-est} and store them
\EndFor
\phase{Inverting bridging functions}
\For{$j$ in $\{c,t,o,b\}$}
    \For{$k \ne j$}
        \State Calculate sample Kendall's Tau: $\widehat{\tau}_{jk}$
        \State Get the appropriate bridging function $F_{jk}$ and plug-in the estimated cutoffs 
        \State Obtain $\widehat{\Sigma}_{jk}=F^{-1}_{jk}(\widehat{\tau}_{jk}) = \textrm{argmin}_{\rho \in (-1,1)} (F_{jk}(\rho) -\widehat{\tau}_{jk})^2$
    \EndFor
\EndFor
\phase{Getting nearest PD correlation matrix}
\State Get the initial estimate of the latent correlation matrix $\Sigma$ as follows:
\begin{equation*}
    \begin{split}
    \widehat {\Sigma} = \begin{pmatrix}
            \hat{\Sigma}_{cc} &  \hat{\Sigma}_{ct} &  \hat{\Sigma}_{co} & \hat{\Sigma}_{cb} \\
            \hat{\Sigma}_{tc} &  \hat{\Sigma}_{tt} &  \hat{\Sigma}_{to} & \hat{\Sigma}_{tb} \\
            \hat{\Sigma}_{oc} &  \hat{\Sigma}_{ot} &  \hat{\Sigma}_{oo} & \hat{\Sigma}_{ob} \\
            \hat{\Sigma}_{bc} &  \hat{\Sigma}_{bt} &  \hat{\Sigma}_{bo} & \hat{\Sigma}_{bb}
        \end{pmatrix}
    \end{split}
\end{equation*}
\State Use \emph{nearPD} \citep{higham2002computing} function in R to find the nearest positive definite correlation matrix of $\widehat{\Sigma}$ as our final estimate. 
\end{algorithmic}
\end{algorithm}
\end{singlespace}



\section{SGC-Reg: Semiparametric Gaussian Copula Regression Model} \label{sec: SGC-Reg}

In this section, we introduce Semiparametric Gaussian Copula Regression Model (SGC-Reg), establish asymptotic theory for the regression estimators, and discuss advantages of SGC-Reg over relevant traditional regression frameworks. 

A classical regression model for a continuous outcome $Y_i$ is typically written as
\begin{equation}
Y_{i} = \sum_{j=1}^{p} X_{ji}\beta_{j}  + \epsilon_i, \quad i=1, \cdots, n
\end{equation}

The simplest for understanding case is when both the outcome and all covariates are standard normal random variables. In that case, the simple linear regression conceptually assumes that both outcomes and predictors are on the same additive scale and tries to explain the variability of an outcome via variability of predictors. Various transformations of outcome/predictors can be used to deal with possible deviations from normality and symmetry. When outcome is not continuous alternative models such as probit, truncated regression, and other probit-like models have been proposed. However, they often lack the interpretability appeal of a simple linear regression model. 

We now introduce \emph{Semiparametric Gaussian Copula Regression for Mixed Data (SGC-Reg)}  and establish key asymptotical results for the estimators of the regression parameters. SGC-Reg operates and connects underlying continuous normal latent variables that generate observed mixed types variables. SGC-Reg can be seen as a unifying alternative to the linear and probit-like regressions for mixed types of outcomes and predictors.  

We define Semiparametric Gaussian Copula Regression for Mixed Data as follows.

\begin{equation}
\left \{
\begin{array}{ll}
\mbox{Observed variables:} & (Y_1, \mathbf{X}_1), \dots, (Y_n, \mathbf{X}_n) \overset{i.i.d}{\sim} GLNPN_{p+1}(0,\Sigma, (f_Y, f_X),\mathbf{\Delta})
\\
\mbox{Latent variables:} & (Z^Y_1, \mathbf{Z}^X_1), \dots, (Z^Y_n, \mathbf{Z}^X_n)\quad
\overset{i.i.d}{\sim} NPN_{p+1}(0,\Sigma,(f_X, f_Y))\\
\\
\mbox{SGC-Reg for latent variables:} & f_{Y}(Z^Y_i) =\sum_{k \in \{c,t,o,b\}} \sum_{j=1}^{p_k} f_{X}(Z_{kji}^X)\beta_{kj} + \epsilon_{i}, \quad i=1,\dots, n.\\
\end{array}\label{eqn:SGC-Reg}
\right .
\end{equation}

Thus, SGC-Reg is a simple linear regression on a latent space with the outcome $f_{Y}(Z^Y_i)$ and predictors $f_{X}(Z^X_i)$, which, according to GLNPN, are jointly normal: $(f_{Y}(Z^Y_i), f_{X}(Z^X_i)) \overset{i.i.d}{\sim} N_{p+1}(0,\Sigma)$ with the correlation matrix $\Sigma$ that assumes the following partition: 

\[\begin{bmatrix}
\Sigma_{YY} & \Sigma_{YX}\\
\Sigma_{XY} & \Sigma_{XX}\\
\end{bmatrix}.\]

In SGC-Reg, we also assume that $\epsilon_i$ are i.i.d. from $N(0, 1- \Sigma_{YX} \Sigma_{XX}^{-1} \Sigma_{XY})$. 

It immediately follows that the regression coefficient $\bm{\beta}= \Sigma_{XX}^{-1} \Sigma_{XY}$. To estimate $\bm{\beta}$, we propose to use the estimate of $\Sigma$ obtained via bridging as described in Section \ref{sec: estimation}. Let $\hat{\Sigma}_n$ be the estimated latent correlation matrix for the model  \eqref{eqn:SGC-Reg}. Then, the estimates of the regression coefficient is given by $\hat{\beta}_n = {\hat{\Sigma}}_{n_{XX}}^{-1} {\hat{\Sigma}}_{n_{XY}}$. 

Previous work on bridging  \citep{fan2017high, quan2018rank, yoon2020sparse} have shown the consistency of the latent correlation estimators. In the next theorem, we prove the stronger result of asymptotic normality for the latent correlation estimators. We, then, build on that and establish the asymptotic normality of the SGC-Reg regression coefficient estimators. This novel contribution enables us to calculate fast asymptotic confidence intervals instead of relying on slow resampling based approaches. To formulate the theorem, we will need the following notations: let $vec(A)$ and $vecl(A)$ denote the vectorized matrix $A$ and vector of lower-triangular elements of matrix $A$, respectively. Thus,  $vecl(\hat{\Sigma}_n)$ and $vec(\hat{\Sigma}_n)$ are vectors of length $\frac{p(p-1)}{2}$ and $p^2$, respectively. 

\begin{theorem}
\label{thm: asymp}
Suppose the following assumptions \citep{eicker1963asymptotic} hold true: (i) the rank of $P_n = \hat{\Sigma}_{nXX}$ is $p$ and (ii) $\frac{\lambda_{\textrm{max}}(P_n)}{n({\lambda_{\textrm{min}}(P_n)})^2} \rightarrow 0$ as $n \rightarrow \infty$, where $\lambda_{\textrm{min}}(\cdot)$ and $\lambda_{\textrm{max}}(\cdot)$ denote the smallest and largest eigenvalues of a matrix respectively. 

Then, 
\begin{itemize}
\item[(i)] $\sqrt{n}(vecl(\hat{\Sigma}_n) - vecl(\Sigma)) \overset{d}{\to} N_{p(p-1)/2}(0, V_{\Sigma}), \mbox{ as } n \to \infty$. 

\item[(ii)] $\sqrt{n}(\hat{\beta}_n - \beta)\overset{d}{\to} N_{p}(0, V_{\beta}), \mbox{ as } n \to \infty$. 
\end{itemize}
\end{theorem}

\begin{proof}
Here, we layout the key ideas of the proof. First, using asymptotics of U-statistics in \cite{hoeffding1992class} and \cite{el2003spearman}, we establish the asymptotic normality of the Kendall's Tau estimates. Since the latent correlations are  deterministic function (inverse bridging function) of the Kendall's Tau correlations, we use the Delta method to obtain the asymptotic normality of the latent correlations. Next, we project the latent correlations onto a space of independent parameters \citep{archakov2018new}, so that we can apply the Delta method to obtain the asymptotic normality of the SGC-Reg regression coefficients. The regularity assumptions ensure the stability of the transformation $\hat{\beta}_n = {\hat{\Sigma}}_{n_{XX}}^{-1} {\hat{\Sigma}}_{n_{XY}}$ of the latent correlation matrix, so that we can apply the Delta method. The detailed proof and analytical expressions of $V_\beta$ and $V_\Sigma$ are provided in \ref{sec: proofs}.
\end{proof}

As part of the derivations, we solved a non-trivial computational problem by developing an efficient way of computing of the asymptotic covariance of Kendall's Tau matrix $V_\Sigma$. We extended the approach of \cite{perreault2022efficient} to calculate the covariance of Kendall's Tau in the case of mixed data and presence of ties. As it is explained in Supplementary Material, our approach requires only $O(n\log{n})$ FLOPs compared to the $O(n^4)$ FLOPs using a naive brute-force approach. This reduction in computational burden enables us a fast calculation of the asymptotic variance for moderate-to-large $n$.


\subsection{Advantages of SGC Framework}
The key advantages of the proposed regression and dimension reduction framework span across the entire spectrum of assumptions, modelling, estimation, interpretation of results, and prediction. Table \ref{tab: comparison} provides a comparison of SGC-based framework vs traditional alternatives. Briefly, the approach offers i) mutually consistent conditional modeling across four different types of outcome, ii) semi-parametric likelihood-free and computationally fast estimation, iii) a natural normalization of mixed types scales, iv) regression and dimension reduction results that are invariant to monotone transformation, v) intuitive interpretation of the regression and dimension reduction results.

\begin{singlespace}
\begin{longtable}{|p{.20\textwidth} | p{.40\textwidth} |p{.40\textwidth}|}
\toprule
\textbf{Aspect}                 & \textbf{Traditional models (Observed space) }                                                               & \textbf{SGC (Latent space)}                                                                        \\ 
\midrule
\endhead 
Conditional associations   & Linear, probit, truncated, and ordinal probit regressions. & Mutually consistent conditional SGC-Reg models.           \\
                           & \emph{Testing conditional independence:} model-specific. 
                           & \emph{Testing conditional dependence:} unified-approach. \\
                                                     & \emph{Goodness of fit}: model-specific                           & \emph{Goodness of fit:} latent R-square  \\
                           \midrule
Estimation                    & Model-specific, often likelihood-based, computationally intensive. Efficient.

& Method of moments, computationally fast. Less efficient.                                  \\                        

                           \midrule
Transformation and Scaling                    & Manual transformation/scaling of outcomes and predictors.  Non-uniform interpretation across  mixed types. 
& Invariance to monotone transformations. Automatic scaling.                                          \\
\midrule
Distributional assumptions & Parametric.                                                               & Semi-parametric.                                      \\
\midrule
Missing data imputation    & Imputation by mean or restricted to complete cases.                                                & Using latent correlation to impute missing data under missing-at-random assumption.                     \\
\midrule
Interpretation             & Original scale interpretation. The signs of regression coefficients define the direction of association. Regression coefficients are not cross-comparable.                                             & Quantile scale interpretation. The signs of regression coefficients define the direction of association. Regression coefficients are standardized and cross-comparable.                                                                      \\
\midrule
Prediction                 & Model-specific.                                                         & Best linear unbiased predictors via conditional expectation.  \\ 
\midrule 
 Handling multicollinearity & No general method to do principal component analysis of mixed data and perform subsequent regression  & Latent principal component analysis allows to perform joint regression using latent PCs and handle multicollinearity \\\bottomrule
 \caption{Comparison between traditional models and SGC-Reg.}
 
\label{tab: comparison}
\end{longtable}
\end{singlespace}

\subsection{SGC-PCA and SGC-PCR: Principal Component Analysis and 
Principal Component Regression}\label{sec:pcr}

To obtain a low-dimensional representation of the latent predictors, we 
define Semiparametric Gaussian Copula Principal Component Analysis 
(SGC-PCA). Let $\Sigma_{XX}$ denote the latent correlation matrix of the 
predictors. The latent principal component loadings $V_m = (v_1, \ldots, 
v_m)$ $(p \times m)$ are defined as the first $m$ eigenvectors of 
$\Sigma_{XX}$, corresponding to eigenvalues $\lambda_1 \geq \cdots \geq 
\lambda_m > 0$. The $i$-th latent principal component score is then given 
by:
$$
PC_{im} = f_X(Z_i^X) V_m, \quad i = 1, \ldots, n.
$$
The proportion of total latent variance explained by the first $m$ 
components is $\sum_{j=1}^m \lambda_j / \sum_{j=1}^p \lambda_j$. In 
practice, $\Sigma_{XX}$ is replaced by its estimate $\hat{\Sigma}_{nXX}$ 
obtained via algorithm \ref{alg: estimation}, 
yielding estimated loadings $\hat{V}_m$ and estimated eigenvalues 
$\hat{\lambda}_1 \geq \cdots \geq \hat{\lambda}_m$.

A key advantage of SGC-PCA over standard PCA applied to observed mixed 
data is that it operates on the latent correlation matrix $\Sigma_{XX}$, 
which places all variables on a common scale and is invariant to monotone 
transformations of the original variables. This yields principal 
components that are interpretable across heterogeneous measurement types 
without requiring ad hoc transformations or standardizations.

To address possible multicollinearity among predictors, we extend SGC-Reg 
to principal component regression by defining a regression model between 
the outcome and the first $m$ latent principal components of the 
predictors. Using the loadings $V_m$ from SGC-PCA, the regression model 
becomes:
$$
f_{Y}(Z^Y_i) = f_{X}(Z_{i}^X) V_m \beta^{PC}_m + \epsilon_{i}, \quad 
i=1,\dots, n,
$$
where $\beta^{PC}_m$ is an $m \times 1$ coefficient vector. To obtain the regression coefficient, observe that the joint distribution 
of $(f_{Y}(Z^Y_i), f_{X}(Z_{i}^X) V_m) \sim N(0, \tilde{\Sigma}_m)$, 
where:
$$
\tilde{\Sigma}_m = \begin{pmatrix} 1 & \Sigma_{YX}V_m\\ V_m^T \Sigma_{XY} 
& \mathrm{diag}(\lambda_1, \cdots, \lambda_m)\end{pmatrix}.
$$
Here, $\mathrm{diag}(\lambda_1, \cdots, \lambda_m)$ contains the first $m$ 
eigenvalues of $\Sigma_{XX}$, which appear as the diagonal covariance 
matrix of the projected predictors $f_X(Z_i^X)V_m$ due to the 
orthogonality of principal component loadings. It immediately follows 
that the SGC-PCR coefficient is:
$$
\beta^{PC}_m = \mathrm{diag}\!\left(\frac{1}{\lambda_1}, \cdots, 
\frac{1}{\lambda_m}\right) V_m^{T} \Sigma_{XY}.
$$
All elements of this expression are obtained directly from the latent 
correlation matrix estimators. In 
practice, $\Sigma_{XY}$, $V_m$, and $\lambda_1, \ldots, \lambda_m$ are 
replaced by their estimates $\hat{\Sigma}_{nXY}$, $\hat{V}_m$, and 
$\hat{\lambda}_1, \ldots, \hat{\lambda}_m$ respectively. We can derive the asymptotic distribution of SGC-PCR regression coefficients using the results in Theorem $\ref{thm: asymp}$. Under the assumption that, $W_m = diag(\frac{1}{\lambda_1}, \cdots, \frac{1}{\lambda_m}) V_m^{T}$ is constant, the asymptotic variance of $\sqrt{n}\beta^{PC}_m$ equals $W_m^T V W_m$, where, $V$ equals the asymptotic variance of $\Sigma_{XY}$. 

\section{Methodological Applications}
GLNPN framework provides a readily available way to perform predictions of latent variables as well as imputation of missing mixed data observations. We develop both approaches in this section.

\subsection{Latent variable predictions} \label{sec: latent-pred}

Although latent variables are not required for estimating regression parameters in SGC-Reg, certain applications of the framework require prediction of latent representations. To address this, we adopt an approach similar to best linear unbiased predictors (BLUPs) in mixed-effects models and use conditional expectations to obtain predictors of latent variables given observed data. At this stage we do not distinguish between outcomes and predictors and therefore omit the participant index $i$, since latent predictions are obtained using only subject-specific observed variables.

Let $L=(L_c,L_t,L_o,L_b)=f(Z_c,Z_t,Z_o,Z_b)$ denote the latent variables obtained through coordinate-wise monotone transformations $f$ as defined in the GLNPN model. We use the notation $L_{-c}=(L_t,L_o,L_b)$ and $L_{-ct}=(L_o,L_b)$, where $ct$ denotes the union of continuous and truncated indices. For a set of indices $a$, $\Sigma_{a,a}$ denotes the corresponding submatrix of $\Sigma$, while $\Sigma_{a,-a}$ and $\Sigma_{-a,a}$ denote cross-covariance blocks.

To compute $E(L|X)=E(L_c,L_t,L_o,L_b|X_c,X_t,X_o,X_b)$, two cases arise depending on whether the truncated variable equals zero.

When $X_t=0$, the continuous component satisfies $L_c=f_c(X_c)$, while the observed values $X_t,X_o,X_b$ restrict the remaining coordinates $L_{-c}$ to lie in a Cartesian product of intervals $B=\{\times_{\lambda \notin c} B_\lambda\}$ determined by the cutoffs. Under the model assumptions $\{X_t=x_t,X_o=x_o,X_b=x_b\}\Longleftrightarrow\{L_{-c}\in B\}$. Using the conditional distribution

\begin{equation}
L_{-c}\mid L_c \sim N\!\left(\Sigma_{-c,c}\Sigma_{c,c}^{-1}L_c,\;
\Sigma_{-c,-c}-\Sigma_{-c,c}\Sigma_{c,c}^{-1}\Sigma_{c,-c}\right),
\end{equation}

we obtain

\begin{equation}
E(L|X_c,X_t,X_o,X_b)=
\big(f_c(X_c),\,E(L_{-c|c}\mid L_{-c|c}\in B_{-c})\big).
\end{equation}

When $X_t>0$, both $L_c=f_c(X_c)$ and $L_t=f_t(X_t)$ are determined through the transformations, while the ordinal and binary components restrict $L_{-ct}=(L_o,L_b)$ to lie in a set $B_{-ct}=\{\times_{\lambda \notin ct}B_\lambda\}$. Using

\begin{equation}
L_{-ct}\mid L_{ct} \sim 
N\!\left(\Sigma_{-ct,ct}\Sigma_{ct,ct}^{-1}L_{ct},\;
\Sigma_{-ct,-ct}-\Sigma_{-ct,ct}\Sigma_{ct,ct}^{-1}\Sigma_{ct,-ct}\right),
\end{equation}

the latent predictor becomes

\begin{equation}
E(L|X_c,X_t,X_o,X_b)=
\big(f_c(X_c),\,f_t(X_t),\,E(L_{-ct|ct}\mid L_{-ct|ct}\in B_{-ct})\big).
\end{equation}

Computing these quantities requires estimates of the monotone transformations, the truncation sets, and the mean of a truncated multivariate normal distribution. Following \cite{liu2009nonparanormal}, the transformations are estimated via empirical CDFs,

\begin{equation}
\hat G_{cn}(x)=\frac{1}{n+1}\sum_{i=1}^n I(X_{ci}\le x), 
\qquad
\hat G_{tn}(x)=\frac{1}{n+1}\sum_{i=1}^n I(X_{ti}\le x), \; x>0,
\end{equation}

with $\hat f_c(x)=\Phi^{-1}(\hat G_{cn}(x))$ and $\hat f_t(x)=\Phi^{-1}(\hat G_{tn}(x))$. The interval sets $B_{-c}$ and $B_{-ct}$ are obtained using the moment-based cutoff estimators from Section~\ref{sec: estimation}. Expectations of truncated multivariate normal distributions are computed using the recursive moment generating function approach of \cite{tmvtnorm}, implemented in the \emph{tmvtnorm} R package. Importantly, latent predictions can be computed independently for each subject, allowing straightforward parallelization.

\subsection{Missing data imputation}
Next, we show how to perform imputation of missing data under GLNPN. 

Suppose we have missing observations for a particular subject. We split the full vector $X$ into observed and missing parts as $X=(X_O,X_M)$, where $O$ denotes observed and $M$ denotes missing parts and subject-specific index $i$ has been omitted for notational simplicity. First we predict $E[L_M|X_O]$ and then obtain the prediction of $X_M$ using an appropriate transformation-then-truncation step applied to $E[L_M|X_O]$. Remember that
\begin{equation}
\begin{aligned}    
L_{M|O}=L_{M}|L_O \sim N(\Sigma_{M,O}\Sigma^{-1}_{O,O}L_O, \Sigma_{M,M}-\Sigma_{M,O}\Sigma_{O,O}^{-1}\Sigma_{O,M})
\end{aligned}
\end{equation}
As $X_O$ is $\sigma(L_O)$-measurable random variable, where $\sigma(L_O)$ denotes the $\sigma$-algebra generated by $L_O$, we can use the tower property of conditional expectations to get the following identity 
\begin{equation}
\begin{aligned}
E[L_M|X_O]=E[E[L_M|L_O]|X_O]= E[\Sigma_{M,O}\Sigma^{-1}_{O,O}L_O|X_O]=\Sigma_{M,O}\Sigma^{-1}_{O,O}E[L_O|X_O]
\end{aligned}
\end{equation}

Finally, we can calculate $E[L_O|X_O]$ from the equation above using the same steps as in previous section.

\section{Simulation study}\label{sec: sim}

We conduct a series of simulation experiments to evaluate the performance of the proposed framework. The data generation mechanism is designed to mimic the dependence structure observed in the NHANES analysis. We begin with the estimated $62 \times 62$ latent correlation matrix obtained in Section~\ref{sec: nhanes} (Figure~\ref{fig:latcor-nhanes}). From this matrix we extract random submatrices of size $p=8,16,24,$ and $32$, and also consider the full matrix with $p=62$. These define five simulation settings corresponding to increasing dimensionality.

For each setting we generate $n=9000$ latent observations, $
(L_{i1},L_{i2},\ldots,L_{ip}) \sim N(0,\Sigma), \qquad i=1,\ldots,n .$ The observed variables $X$ are then constructed using the GLNPN model by partitioning the $p$ variables into approximately equal groups of binary, ordinal, truncated, and continuous types. Specifically, the first $\lfloor p/4 \rfloor$ variables are binary, the next $\lfloor p/4 \rfloor$ are ordinal (equally divided between three- and four-level variables), the next $\lfloor p/4 \rfloor$ are truncated, and the remaining variables are continuous. Random cutoff values are generated to define the discretization and truncation mechanisms and are kept fixed across simulation settings.

For each simulated dataset we estimate the latent correlation matrix using Algorithm~\ref{alg: estimation}. The first binary variable is treated as the outcome, and the latent regression coefficient is estimated accordingly. Asymptotic confidence intervals are constructed using Theorem~\ref{thm: asymp}, and coverage is assessed by checking whether the true coefficient lies within the interval. This procedure is repeated for $500$ independent replicates in each simulation setting to evaluate empirical coverage and estimation accuracy.

\begin{figure}[H]
\centering
\begin{subfigure}[b]{  0.95\textwidth}
\centering
\includegraphics[scale=0.3]{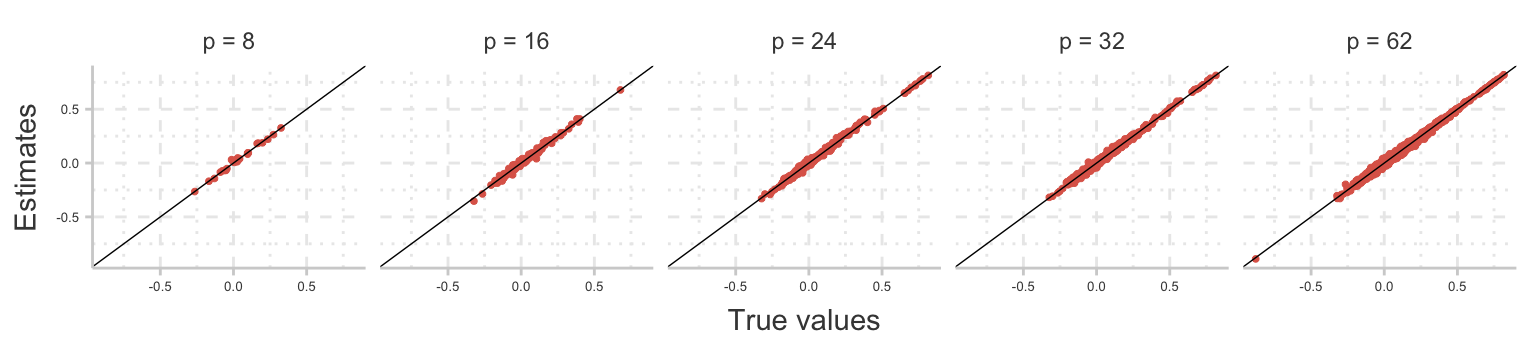}
\caption{Latent correlation}
\label{fig:sim-r-est}
\end{subfigure}
     \hfill
\centering
\begin{subfigure}[b]{  0.95\textwidth}
\centering
\includegraphics[scale=0.3]{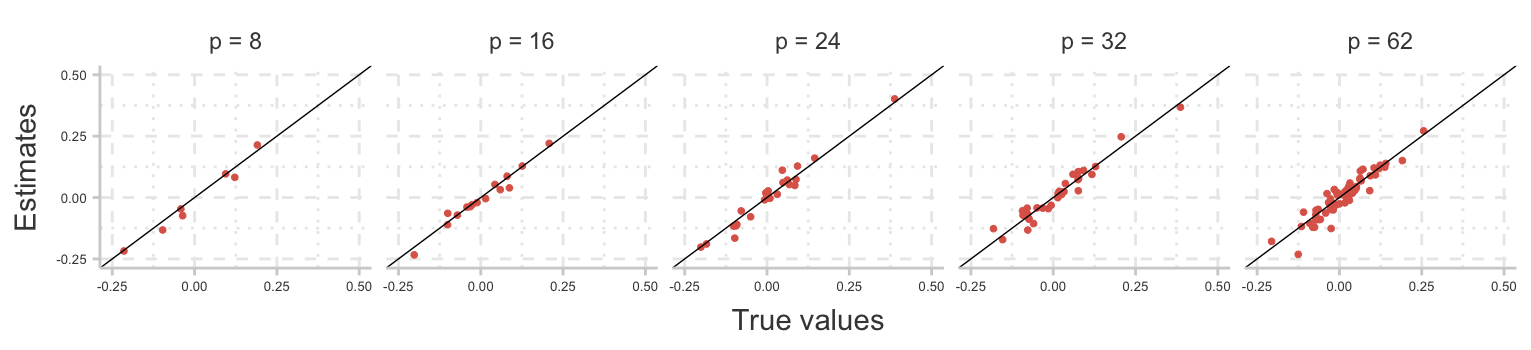}
\caption{The latent regression coefficients}
\label{fig:sim-beta-est}
\end{subfigure}
     \hfill
     \caption{Comparison of true values and estimates for one random replicate in each of the simulation settings. The black line denotes $y=x$ line.}
         \label{fig:sim-est}
\end{figure}

Figure \ref{fig:sim-est} compares the estimated latent correlations and latent regression coefficients with their true values. Across all simulation settings, the estimates align closely with the $y=x$ line, indicating good estimation accuracy. Figure \ref{fig:sim-cov} reports the coverage of the proposed asymptotic confidence intervals. The median coverage for both latent correlations and regression coefficients in the SGC-Reg model exceeds 0.95 across all settings. Even in the lowest-dimensional case (Setting A, $p=8$), the empirical coverage is already close to the nominal level.

As the dimension increases, the number of latent correlation parameters grows at rate $O(p^2)$, reaching 120 even for a moderate dimension of $p=16$. The slight over-coverage observed in higher-dimensional settings is likely due to using a fixed number of simulation replicates (500) for all values of $p$. When the number of parameters increases, a larger number of replicates would provide a more stable Monte Carlo estimate of coverage. Due to computational constraints we restrict attention to 500 replicates per setting, although a larger simulation study would allow a more precise assessment of coverage.

\section{Frailty-related deficits and 5-year mortality in older adults from NHANES, 1999--2010}
\label{sec: nhanes}

Frailty is a multidimensional geriatric construct reflecting loss of physiologic reserve and increased vulnerability to adverse outcomes such as disability, hospitalization, and mortality. In the deficit-accumulation framework, frailty is summarized by combining abnormalities across multiple domains, including chronic diseases, functional limitations, self-rated health, and laboratory biomarkers \citep{mitnitski2001accumulation, rockwood2007frailty, searle2008standard}. A common operationalization is the Frailty Index (FI), which aggregates such deficits into a single summary score and has been widely used as a marker of biological aging and as a predictor of mortality \citep{mitnitski2001accumulation, searle2008standard, blodgett2015frailty}. While highly useful for risk stratification, the FI largely treats deficits as exchangeable and does not explicitly model their dependence structure; in addition, ordinal and continuous variables are often coarsened or dichotomized in practice. These features motivate a richer multivariate analysis that preserves measurement scale and clarifies how different frailty domains jointly relate to mortality.

NHANES is particularly well suited for this application because it combines interview, examination, and laboratory data in a complex, multistage probability sample of the U.S. civilian noninstitutionalized population, and these records can be linked to mortality follow-up through the National Center for Health Statistics (NCHS) Linked Mortality Files. Table \ref{tab:variable-list} lists the 61 frailty-related variables used in our analysis: 29 continuous laboratory and examination measures, 17 ordinal self-reported or functional variables, and 15 binary disease indicators. Examples include glycohemoglobin, creatinine, and red blood cell indices among the continuous variables; self-rated health, healthcare use, and graded difficulty with physical tasks among the ordinal variables; and stroke, cancer, diabetes, and coronary heart disease among the binary variables. Taken together, these variables reflect frailty as multisystem deficit accumulation rather than dysfunction within a single organ system.

Our analytic sample pools NHANES waves from 1999--2010 and includes participants older than 60 years with mortality follow-up information. After excluding subjects with more than $20\%$ missingness, 9,478 participants remained for analysis. Because NHANES combines interview, examination, and laboratory measurements across multiple survey waves, this pooled dataset provides unusually rich information on frailty-related deficits spanning clinical conditions, functional limitations, and physiological biomarkers. The resulting multivariate dataset is therefore well suited for studying the joint dependence structure among frailty-related variables and their association with 5-year mortality.

\begin{figure}[H]
\centering
\includegraphics[scale=0.14]{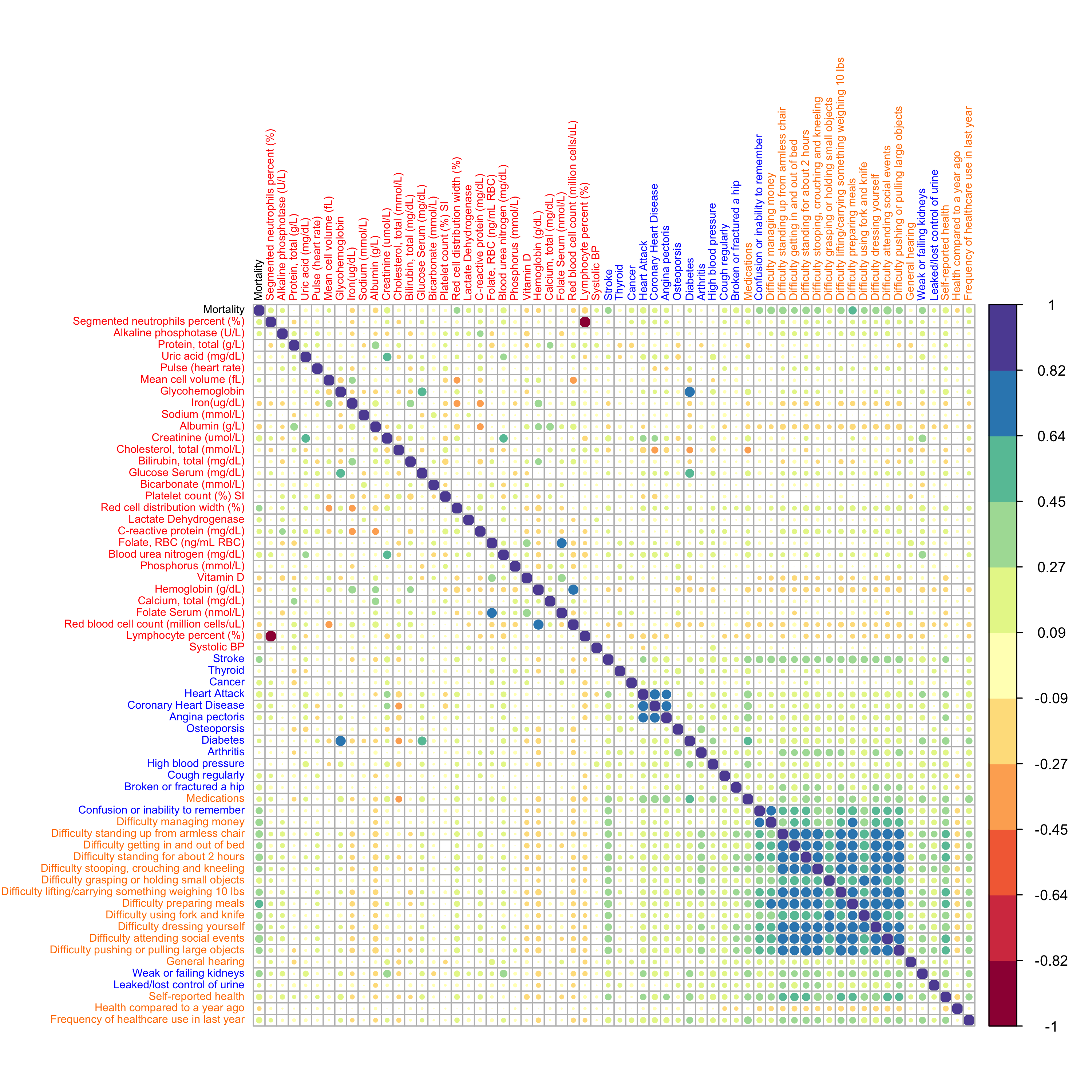}
\caption{The estimated latent correlation matrix of mortality and other frailty variables from NHANES 1999-2010 (subjects older than 60 years)}
\label{fig:latcor-nhanes}
\end{figure}

We first use the semiparametric Gaussian copula framework to estimate the latent dependence structure among mortality and the frailty-related variables. Figure \ref{fig:latcor-nhanes} displays the estimated latent correlation matrix and reveals several clinically interpretable patterns. The strongest block occurs among the self-reported difficulty variables, highlighting functional decline as a core dimension of frailty. A second cluster appears among cardiovascular conditions, especially heart attack, coronary heart disease, and angina pectoris. More moderate local structure is also visible among laboratory markers, indicating that frailty-related information is distributed across functional, clinical, and biomarker domains. A key advantage of the SGC approach is that it places all pairwise associations on a common latent scale despite the mixed measurement types.

We next examine the marginal association between each frailty variable and 5-year mortality using separate SGC-Reg models. Figure \ref{fig:beta-marg-ci} shows that the strongest marginal signals are dominated by functional and self-reported deficits, including difficulty preparing meals, difficulty attending social events, difficulty managing money, difficulty standing for about 2 hours, difficulty standing up from an armless chair, difficulty lifting or carrying something weighing 10 pounds, and difficulty dressing oneself. Kidney-related and cognitive variables, such as weak or failing kidneys and confusion or inability to remember, also rank among the strongest marginal predictors. From a frailty perspective, this pattern is biologically plausible: impairments in daily activities, cognition, and organ-specific function often reflect broad loss of reserve across multiple systems and therefore show strong unadjusted associations with survival.

\begin{figure}[H]
\centering
\begin{subfigure}[b]{  0.95\textwidth}
\centering
\includegraphics[scale=0.25]{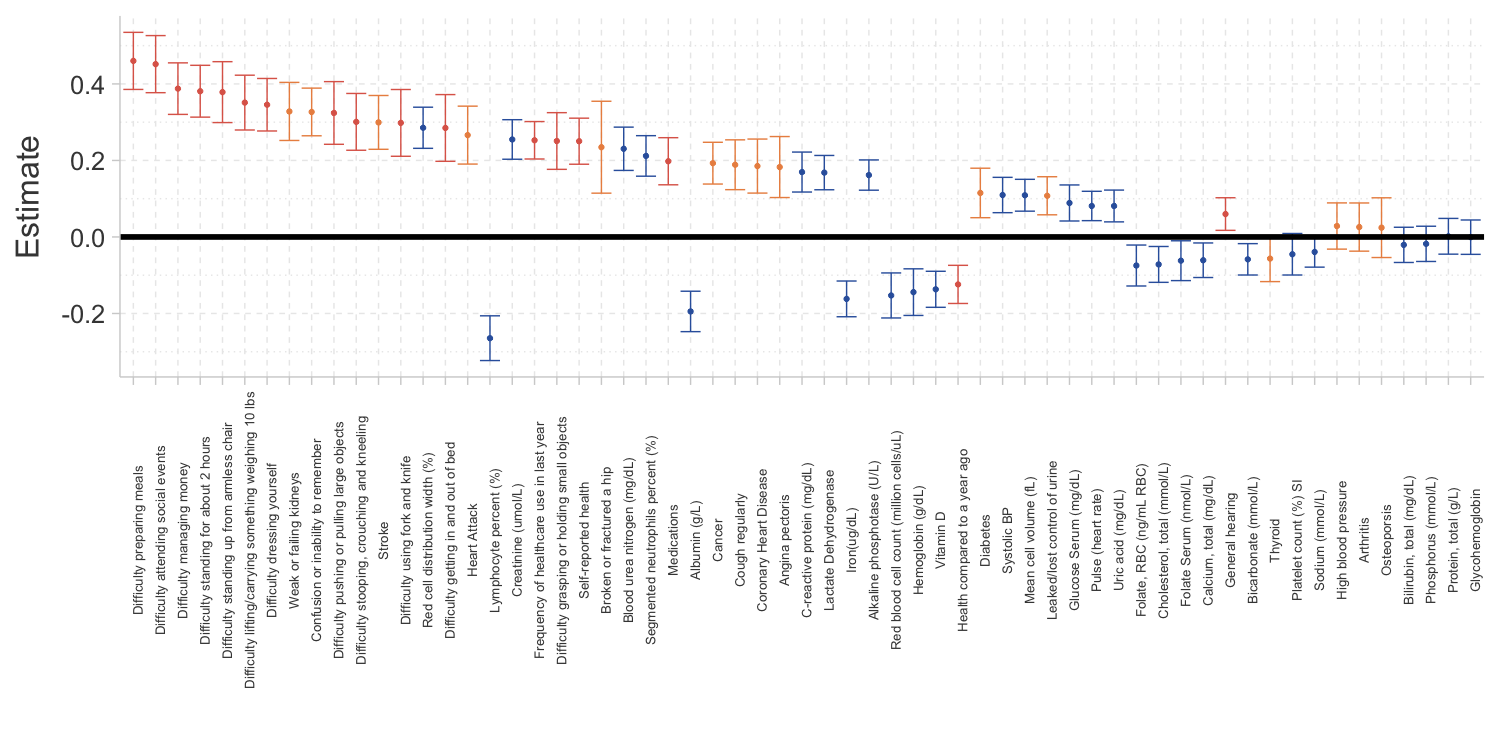}
\caption{The marginal latent regression coefficients of mortality with frailty variables (ordered with respect to absolute values)}
\label{fig:beta-marg-ci}
\end{subfigure}
     \hfill
\centering
\begin{subfigure}[b]{  0.95\textwidth}
\centering
\includegraphics[scale=0.25]{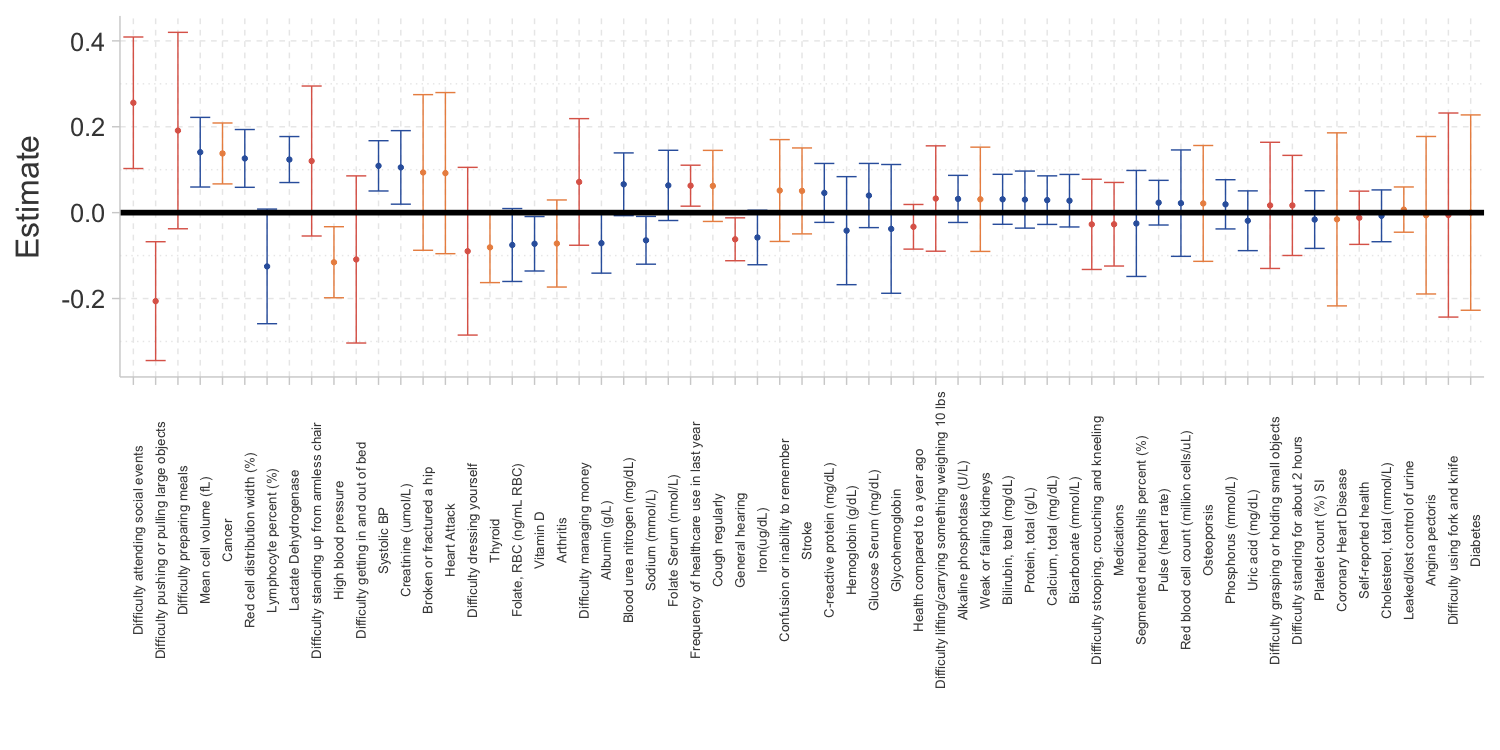}
\caption{The joint latent regression coefficients of mortality with frailty variables (ordered with respect to absolute values)}
\label{fig:beta-joint-ci}
\end{subfigure}
     \hfill
     \caption{Comparison of marginal and joint latent regression coefficients between mortality and frailty variables}
         \label{fig:beta}
\end{figure}

To distinguish shared from variable-specific effects, we then fit the full joint latent regression model, SGC-Reg, with mortality as the outcome and all frailty variables entered simultaneously. Figure \ref{fig:beta-joint-ci} shows substantial attenuation of many marginal effects after joint adjustment, with several coefficients shrinking toward zero and a few changing sign. This is expected in a frailty setting, where many observed deficits are overlapping manifestations of the same underlying vulnerability. In particular, the sign change for difficulty pushing or pulling large objects is consistent with multicollinearity induced by the strong dependence among the functional difficulty variables. The joint SGC-Reg analysis is therefore useful not only for ranking predictors, but also for identifying which variables retain distinct information after accounting for the broader latent frailty structure.

Our framework addresses multicollinearity through semiparametric Gaussian copula principal component analysis (SGC-PCA) and semiparametric Gaussian copula principal component regression (SGC-PCR) (Section \ref{sec:pcr}). By summarizing the latent correlation structure into a small number of orthogonal components, SGC-PCA provides a clinically interpretable representation of the major frailty domains, while SGC-PCR relates those domains to mortality in a dimension-reduced regression model. Figure \ref{fig:pcloading-nhanes} shows that the leading latent principal components are readily interpretable: PC$1$ explains $19\%$ of the total variation and is dominated by self-reported difficulty variables, indicating a broad functional limitation axis; PC$2$ ($6\%$) is primarily associated with heart disease variables; PC$3$ ($5\%$) is driven by diabetes-related deficits; PC$4$ ($5\%$) is dominated by hemoglobin-related laboratory measures; PC$5$ and PC$6$ together capture additional hematologic variation, especially white blood cell related measures; and PC$7$ ($3\%$) reflects kidney-related function. This structure is consistent with the view of frailty as a multisystem process involving disability, disease burden, and physiologic dysregulation rather than any single deficit. Performing SGC-PCR on these latent components allows us to assess which broad biological or functional domains are most strongly associated with 5-year mortality. Among the first five components, PC$1$ (functional difficulty), PC$4$ (red blood cell related measures), and PC$5$ (white blood cell related measures) are significantly associated with mortality (Figure \ref{fig:latpcreg-nhanes}), suggesting that both overt functional impairment and subtler hematologic dysregulation contribute meaningfully to mortality risk in older adults.

\section{Discussion}
The main contribution of this paper is a joint modeling framework for mixed data types based on the semiparametric Gaussian copula (SGC). The proposed regression method, SGC-Reg, performs linear regression in the latent SGC space and yields mutually consistent conditional regression models for continuous, truncated, ordinal, and binary outcomes. Unlike likelihood-based copula models, the proposed approach is likelihood-free, scale-invariant, and computationally efficient. By transforming variables to a common latent scale, SGC-Reg automatically standardizes the covariates and produces regression coefficients that are directly comparable across different data types. In addition, the framework can provide natural $R^2$ analogs for mixed-type outcomes and enable methodological extensions, such as latent-variable prediction and missing-data imputation.

The NHANES application illustrates how the proposed framework can be used to study the dependence structure among frailty-related deficits and their association with mortality in a large epidemiologic dataset. Although NHANES is designed to provide nationally representative estimates through a complex sampling scheme, our analysis focuses on the pooled analytic sample and does not incorporate survey weights. Consequently, the reported associations should be interpreted as relationships within the study sample rather than as design-based population estimates. Extending the SGC framework to incorporate survey weights would allow population-level inference and represents an interesting direction for future methodological development.

From a computational perspective, the proposed approach scales well to moderately high dimensions. Estimation requires calculating $O(p^2)$ pairwise latent correlations, and the computation of Kendall's Tau for each pair requires only $O(n\log n)$ floating point operations. Thus the overall complexity grows quadratically in the number of variables. In addition, we proposed a computationally efficient procedure for estimating the asymptotic variance-covariance matrix of the SGC-Reg regression estimates, reducing the computational cost from $O(n^4)$ under a naive implementation to $O(n\log n)$.

Several limitations of the current framework suggest directions for future work. First, SGC-Reg assumes a uniform direction and magnitude of the latent effect for each variable and therefore may be less flexible than models that allow level-specific effects for ordinal outcomes. More flexible GLNPN-type mechanisms involving multiple latent variables could be developed to capture non-uniform effects across outcome levels. Second, although the method provides interpretable regression coefficients on the latent scale, interpreting results in terms of the original measurement scale is less direct. Developing quantile-scale interpretations of SGC-Reg regression coefficients would therefore be valuable.

Future research directions are numerous. The SGC framework could be extended to accommodate survival outcomes, which would be particularly relevant for epidemiologic applications such as mortality analyses. Extensions to functional data, longitudinal data, or multilevel mixed-type data structures would also be of substantial interest. In addition, the latent representations produced by the SGC model provide a natural starting point for extending clustering, network analysis, and other distance-based methods to heterogeneous mixed-type datasets.
\section{Software}

A \verb|R| package implementation of the latent correlation estimation has been made available as \verb|fromXtoRMixed| function as part of the \verb|SGCTools| package (\url{https://github.com/Ddey07/SGCTools}).  

\newpage

\appendix

\section{Proofs}
\label{sec: proofs}
\textbf{Proof of Theorem \ref{thm:bridge}}:
Before we start the proof, let's get familiar with some notations. Let us denote $\Tilde{\Phi}_2((a,b),(c,d),\rho)$ as the probability of the rectangle $\{(u,v): a < u <c, b < v < d\}$ for a standard bivariate normal with correlation $\rho$. And, for $d>2$, $\Tilde{\Phi}_d((a_1,a_2, \cdots, a_d),(b_1, b_2, \cdots, b_d),S(\rho)$ denotes the probability of the hyper-rectangle $\{(u_1, u_2, \cdots, u_d): a_i < u_i < b_i, i=1, 2, \cdots, d\}$. 

Let $X_j, X_k$ be ordinal with levels $\{0,1,\cdots,l_{j}-1\}$ and $\{0,1,\cdots,l_{k}-1\}$ respectively and $(L_j, L_k)$ is the corresponding latent standard bivariate normal with correlation $\rho$. Then for two independent observations $i,i'$ - 

\begin{align*}
& P(X_{ij} > X_{i'j}, X_{ik} > X_{i'k}) \\
& = \sum_{r=1}^{l_{j}-1}\sum_{s=1}^{l_{s}-1} [P(X_{ij}=r, X_{ik}=s) P(X_{i'j}<r, X_{i'k}<s)]\\
& = \sum_{r=1}^{l_{j}-1}\sum_{s=1}^{l_{s}-1} [P(\Delta_{jr} \leq L_{ij} <\Delta_{j(r+1)}, \Delta_{ks} \leq L_{ik} <\Delta_{(k+1)s}) P(L_{i'j}<\Delta_{jr}, L_{i'k}<\Delta_{ks})]\\
& = \sum_{r=1}^{l_j-1}\sum_{s=1}^{l_k-1}[\Tilde{\Phi}_2((\Delta_{jr},\Delta_{ks}),(\Delta_{j(r+1)},\Delta_{k(s+1)};\rho) \Tilde{\Phi}_2((-\infty,-\infty),(\Delta_{jr},\Delta_{ks});\rho)]
\end{align*}
Similarly, 
\begin{align*}
& P(X_{ij} > X_{i'j}, X_{ik} < X_{i'k})\\
&= \sum_{r=1}^{l_{j}-1}\sum_{s=1}^{l_{s}-1} [P(X_{ij}=r, X_{ik}=s-1) P(X_{i'j}<r, X_{i'k} > (s-1))]\\
& = \sum_{r=1}^{l_{j}-1}\sum_{s=1}^{l_{s}-1} [P(\Delta_{jr} \leq L_{ij} <\Delta_{j(r+1)}, \Delta_{(k-1)s} \leq L_{ik} <\Delta_{ks}) P(L_{i'j}<\Delta_{jr}, L_{i'k}>\Delta_{ks})]\\
& = \sum_{r=1}^{l_{j}-1}\sum_{s=1}^{l_{s}-1} [P(\Delta_{jr} \leq L_{ij} <\Delta_{j(r+1)}, \Delta_{(k-1)s} \leq L_{ik} <\Delta_{ks}) P(L_{i'j}<\Delta_{jr}, -L_{i'k} < - \Delta_{ks})]\\
& = \sum_{r=1}^{l_j-1}\sum_{s=1}^{l_k-1}[\Tilde{\Phi}_2((\Delta_{jr},\Delta_{k(s-1)}),(\Delta_{j(r+1)},\Delta_{ks};\rho) \Tilde{\Phi}_2((-\infty,-\infty),(\Delta_{jr},-\Delta_{ks});-\rho)]
\end{align*}

By symmetry, the population Kendall's Tau, $\tau_{jk}$ for $X_j, X_k$ can be written as follows - 

\begin{align*}
\tau_{jk} = &  2(P(X_{ij} > X_{i'j}, X_{ik} > X_{i'k})- P(X_{ij} > X_{i'j}, X_{ik} < X_{i'k}))\\
= & 2(\sum_{r=1}^{l_j-1}\sum_{s=1}^{l_k-1}[\Tilde{\Phi}_2((\Delta_{jr},\Delta_{ks}),(\Delta_{j(r+1)},\Delta_{k(s+1)};\rho)\Tilde{\Phi}_2((-\infty,-\infty),(\Delta_{jr},\Delta_{ks});\rho)\\
&-\Tilde{\Phi}_2((\Delta_{jr},\Delta_{k(s-1)}),(\Delta_{j(r+1)},\Delta_{ks};\rho) \Tilde{\Phi}_2((-\infty,-\infty),(\Delta_{jr},-\Delta_{ks});-\rho)])\\
= & 2 \sum_{r=1}^{l_j-1}\sum_{s=1}^{l_k-1}[\{\Phi_2(\Delta_{jr},\Delta_{ks}; \rho) + \Phi_2(\Delta_{j(r+1)},\Delta_{k(s+1)};\rho) - \\ 
&\Phi_2(\Delta_{jr},\Delta_{k(s+1)}; \rho) - \Phi_2(\Delta_{j(r+1)},\Delta_{ks}; \rho)\}\Phi_2(\Delta_{jr},\Delta_{ks};\rho) - \\
& \{\Phi_2(\Delta_{jr},\Delta_{k(s-1)}; \rho) + \Phi_2(\Delta_{j(r+1)},\Delta_{ks};\rho) - \\ 
& \Phi_2(\Delta_{jr},\Delta_{ks}; \rho) - \Phi_2(\Delta_{j(r+1)},\Delta_{k(s-1)}; \rho) \} (\Phi(\Delta_{jr})- \Phi_2(\Delta_{jr},\Delta_{ks}; \rho))]\\
= & 2 \sum_{r=1}^{l_j-1}\sum_{s=1}^{l_k-1}[\Phi_2(\Delta_{jr},\Delta_{ks}; \rho)\{\Phi_2(\Delta_{jr},\Delta_{ks}; \rho) + \Phi_2(\Delta_{j(r+1)},\Delta_{k(s+1)};\rho) - \\
& \Phi_2(\Delta_{jr},\Delta_{k(s+1)}; \rho) - \Phi_2(\Delta_{j(r+1)},\Delta_{ks}; \rho)\\
& + \Phi_2(\Delta_{jr},\Delta_{k(s-1)}; \rho) + \Phi_2(\Delta_{j(r+1)},\Delta_{ks};\rho) -  \Phi_2(\Delta_{jr},\Delta_{ks}; \rho) - \Phi_2(\Delta_{j(r+1)},\Delta_{k(s-1)}; \rho)\}] -  \\
& 2 \sum_{r=1}^{l_j-1}\Phi({\Delta_{jr}})\sum_{s=1}^{l_k-1}[\Phi_2(\Delta_{jr},\Delta_{k(s-1)}; \rho) + \Phi_2(\Delta_{j(r+1)},\Delta_{ks};\rho) - \\ 
& \Phi_2(\Delta_{jr},\Delta_{ks}; \rho) - \Phi_2(\Delta_{j(r+1)},\Delta_{k(s-1)}; \rho)] \\
= & 2 \sum_{r=1}^{l_j-1}\sum_{s=1}^{l_k-1}[\Phi_2(\Delta_{jr},\Delta_{ks}; \rho)\{ \Phi_2(\Delta_{j(r+1)},\Delta_{k(s+1)};\rho) -  \Phi_2(\Delta_{jr},\Delta_{k(s+1)}; \rho) \\
& + \Phi_2(\Delta_{jr},\Delta_{k(s-1)}; \rho) - \Phi_2(\Delta_{j(r+1)},\Delta_{k(s-1)}; \rho)\}] - \\
& 2 \sum_{r=1}^{l_j-1}\Phi({\Delta_{jr}})[ - \Phi_2(\Delta_{jr},\Delta_{k(l_k-1)}; \rho) + \Phi_2(\Delta_{j(r+1)},\Delta_{k(l_k-1)}; \rho)]\\
= & 2 \sum_{r=1}^{l_j-1}\sum_{s=1}^{l_k-1}[\Phi_2(\Delta_{jr},\Delta_{ks}; \rho)\{ \Phi_2(\Delta_{j(r+1)},\Delta_{k(s+1)};\rho) - \Phi_2(\Delta_{j(r+1)},\Delta_{k(s-1)}; \rho)\}] + 
\\
& 2 \sum_{r=1}^{l_j-1}[\sum_{s=1}^{l_k-1}\Phi_2(\Delta_{jr},\Delta_{ks}; \rho) \Phi_2(\Delta_{jr},\Delta_{k(s-1)}; \rho) - \Phi_2(\Delta_{jr},\Delta_{ks}; \rho) \Phi_2(\Delta_{jr},\Delta_{k(s+1)}; \rho)]- \\
& 2 \sum_{r=1}^{l_j-1}\Phi({\Delta_{jr}})[ - \Phi_2(\Delta_{jr},\Delta_{k(l_k-1)}; \rho) + \Phi_2(\Delta_{j(r+1)},\Delta_{k(l_k-1)}; \rho)]\\
= & 2 \sum_{r=1}^{l_j-1}\sum_{s=1}^{l_k-1}[\Phi_2(\Delta_{jr},\Delta_{ks}; \rho)\{ \Phi_2(\Delta_{j(r+1)},\Delta_{k(s+1)};\rho) - \Phi_2(\Delta_{j(r+1)},\Delta_{k(s-1)}; \rho)\}] + \\
& 2 \sum_{r=1}^{l_j-1}[\Phi_2(\Delta_{jr},\Delta_{k1}; \rho) \Phi_2(\Delta_{jr},\Delta_{k0}; \rho) - \Phi_2(\Delta_{jr},\Delta_{k(l_k-1)}; \rho) \Phi_2(\Delta_{jr},\Delta_{k(l_k)}; \rho)]- \\
& 2 \sum_{r=1}^{l_j-1}\Phi({\Delta_{jr}})[ - \Phi_2(\Delta_{jr},\Delta_{k(l_k-1)}; \rho) + \Phi_2(\Delta_{j(r+1)},\Delta_{k(l_k-1)}; \rho)] \\
= & 2 \sum_{r=1}^{l_j-1}\sum_{s=1}^{l_k-1}[\Phi_2(\Delta_{jr},\Delta_{ks}; \rho)\{ \Phi_2(\Delta_{j(r+1)},\Delta_{k(s+1)};\rho) - \Phi_2(\Delta_{j(r+1)},\Delta_{k(s-1)}; \rho)\}] + \\
& 2 \sum_{r=1}^{l_j-1}[- \Phi_2(\Delta_{jr},\Delta_{k(l_k-1)}; \rho) \Phi(\Delta_{jr})]- \\
& 2 \sum_{r=1}^{l_j-1}\Phi({\Delta_{jr}})[ - \Phi_2(\Delta_{jr},\Delta_{k(l_k-1)}; \rho) + \Phi_2(\Delta_{j(r+1)},\Delta_{k(l_k-1)}; \rho)] \\
= & 2 \sum_{r=1}^{l_j-1}\sum_{s=1}^{l_k-1}[\Phi_2(\Delta_{jr},\Delta_{ks}; \rho)\{ \Phi_2(\Delta_{j(r+1)},\Delta_{k(s+1)};\rho) - \Phi_2(\Delta_{j(r+1)},\Delta_{k(s-1)}; \rho)\}] - \\
& 2 \sum_{r=1}^{l_j-1}\Phi({\Delta_{jr}})\Phi_2(\Delta_{j(r+1)},\Delta_{k(l_k-1)}; \rho) \\
\end{align*}

The third equality follows from the fact that - $\Tilde{\Phi}_2((a,b),(c,d); \rho) ={\Phi}_2((a,b); \rho) + {\Phi}_2((c,d); \rho)  - {\Phi}_2((a,d); \rho)  - {\Phi}_2((c,b); \rho)$. The seventh equality comes from the following result - for a sequence $\{a_s\}_{s=0}^{n}, \sum_{s=1}^{n-1} (a_s a_{s-1} - a_{s}a_{s+1}) = (a_1 a_0 - a_{n-1} a_{n})$. 

Now, reducing the above calculations for $l_{k}=2$, will yield the bridging function between a general ordinal and binary pairs. 

Now, suppose we have a truncated variable $X_m$ with cutoff $\Delta_m$ and corresponding latent normal variable $L_m$. For consistency in bridging function expressions, we slightly abuse notations to denote $\rho$ as the latent correlation between the truncated and the ordinal variable. Then similar to the ordinal-ordinal case, the calculations will look like - 

\begin{align}
& P(X_{ij} > X_{i'j}, X_{im} > X_{i'm}) \nonumber  \\
& = \sum_{r=1}^{l_{j}-1}[P(X_{ij}=r, X_{im}>0)P(X_{i'j} < r, X_{i'm}=0) \nonumber\\
& + P(X_{ij}=r, X_{i'j} < r, X_{i'm} >0,X_{im} - X_{i'm}>0)] \nonumber\\
& = \sum_{r=1}^{l_{j}-1}[P(\Delta_{jr} \leq L_{ij} <\Delta_{j(r+1)}, L_{im}>\Delta_{m}) P(L_{i'j}<\Delta_{jr}, L_{i'm}\leq \Delta_{m}) \nonumber\\
& + P(\Delta_{jr} \leq L_{ij} <\Delta_{j(r+1)},L_{i'j}<\Delta_{jr}, L_{i'm} > \Delta_{m}, , \frac{L_{im}- L_{i'm}}{2} > 0) ] \nonumber\\
&=\sum_{r=1}^{l_{j}-1}[P(\Delta_{jr} \leq L_{ij} <\Delta_{j(r+1)}, -L_{im} < -\Delta_{m}) P(L_{i'j}<\Delta_{jr}, L_{i'm}\leq \Delta_{m}) \nonumber \\
& + P( - \Delta_{j(r+1)} \leq -L_{ij} <- \Delta_{jr},L_{i'j}<\Delta_{jr}, L_{i'm} > \Delta_{m}, , \frac{L_{i'm}- L_{im}}{\sqrt{2}} < 0) ] \nonumber \\
& = \sum_{r=1}^{l_{j}-1}[\Tilde{\Phi}_2((\Delta_{jr},-\infty),(\Delta_{j(r+1)},-\Delta_m);-\rho) \Tilde{\Phi}_2((-\infty,-\infty),(\Delta_{jr},\Delta_{m});\rho) + \nonumber \\
 & \Tilde{\Phi}_4((-\Delta_{j(r+1)}, -\infty, \Delta_m, -\infty), (-\Delta_{jr}, \Delta_{jr}, \infty, 0); S_{5a}(\rho))] \nonumber \\
 & = \sum_{r=1}^{l_{j}-1}[\Tilde{\Phi}_2((\Delta_{jr},-\infty),(\Delta_{j(r+1)},-\Delta_m);-\rho) \Tilde{\Phi}_2((-\infty,-\infty),(\Delta_{jr},\Delta_{m});\rho) \nonumber \\
 & + \Tilde{\Phi}_3((-\Delta_{j(r+1)}, -\infty, -\infty), (-\Delta_{jr}, \Delta_{jr}, 0); \begin{pmatrix} 1 & 0 & \frac{\rho}{\sqrt{2}} \\ 0 & 1 & \frac{\rho}{\sqrt{2}} \\ \frac{\rho}{\sqrt{2}}  & \frac{\rho}{\sqrt{2}} & 1 \end{pmatrix}) - \nonumber \\
 & \Tilde{\Phi}_4((-\Delta_{j(r+1)}, -\infty, -\infty, -\infty), (-\Delta_{jr}, \Delta_{jr}, \Delta_m, 0); S_{5a}(\rho))]\nonumber\\
 & = \sum_{r=1}^{l_{j}-1}[T_{1r}^{(1)} + T_{2r}^{(1)} - T_{3r}^{(1)}] \label{eq:trunc-deriv-00} \\
 & P(X_{ij} < X_{i'j}, X_{im} > X_{i'm}) \nonumber \\
 & = \sum_{r=1}^{l_{j}-1}[P(X_{ij}=(r-1), X_{im}>0)P(X_{i'j} > (r-1), X_{i'm}=0)\nonumber \\
 & + P(X_{ij}=(r-1), X_{i'j} > (r-1), X_{i'm} >0,X_{im} - X_{i'm}>0)] \nonumber \\
& = \sum_{r=1}^{l_{j}-1}[P(\Delta_{j(r-1)} \leq L_{ij} <\Delta_{jr}, L_{im}>\Delta_{m}) P(L_{i'j} > \Delta_{jr}, L_{i'm}\leq \Delta_{m}) \nonumber \\
& + P(\Delta_{j(r-1)} \leq L_{ij} <\Delta_{jr},L_{i'j} > \Delta_{jr}, L_{i'm} > \Delta_{m}, , \frac{L_{im}- L_{i'm}}{\sqrt{2}} > 0) ] \nonumber \\
&=\sum_{r=1}^{l_{j}-1}[P(\Delta_{j(r-1)} \leq L_{ij} <\Delta_{jr}, -L_{im} < - \Delta_{m}) P(-L_{i'j} < -\Delta_{jr}, L_{i'm}\leq \Delta_{m}) \nonumber \\
& + P(-\Delta_{(jr} \leq -L_{ij} < -\Delta_{j(r-1)}, L_{i'j} > \Delta_{jr}, L_{i'm} > \Delta_{m}, , \frac{L_{i'm}- L_{im}}{\sqrt{2}} < 0) ] \nonumber \\
& =  \sum_{r=1}^{l_{j}-1}[\Tilde{\Phi}_2((\Delta_{j(r-1)},-\infty),(\Delta_{jr},-\Delta_m);-\rho) \Tilde{\Phi}_2((-\infty,-\infty),(-\Delta_{jr},\Delta_{m});-\rho) + \nonumber \\
& \Tilde{\Phi}_4((-\Delta_{jr}, \Delta_{jr}, \Delta_m, -\infty), (-\Delta_{j(r-1)}, \infty, \infty, 0); S_{5a}(\rho))] \nonumber \\
& =  \sum_{r=1}^{l_{j}-1}[\Tilde{\Phi}_2((\Delta_{j(r-1)},-\infty),(\Delta_{jr},-\Delta_m);-\rho) \Tilde{\Phi}_2((-\infty,-\infty),(-\Delta_{jr},\Delta_{m});-\rho) +  \nonumber \\
& \Tilde{\Phi}_3((-\Delta_{jr}, \Delta_m, -\infty), (-\Delta_{j(r-1)}, \infty, 0); S_{6b}(\rho)) - \Tilde{\Phi}_4((-\Delta_{jr}, -\infty, \Delta_m, -\infty), (-\Delta_{j(r-1)}, \Delta_{jr}, \infty, 0); S_{5a}(\rho))] \nonumber \\
& =  \sum_{r=1}^{l_{j}-1}[\Tilde{\Phi}_2((\Delta_{j(r-1)},-\infty),(\Delta_{jr},-\Delta_m);-\rho) \Tilde{\Phi}_2((-\infty,-\infty),(-\Delta_{jr},\Delta_{m});-\rho) + \nonumber \\
& \Tilde{\Phi}_2((-\Delta_{jr}, -\infty), (-\Delta_{j(r-1)}, 0); \frac{\rho}{\sqrt{2}}) - \Tilde{\Phi}_3((-\Delta_{jr}, -\infty, -\infty), (-\Delta_{j(r-1)}, \Delta_m, 0); S_{6b}(\rho)) - \nonumber  \\
& \Tilde{\Phi}_3((-\Delta_{jr}, -\infty, -\infty), (-\Delta_{j(r-1)}, \Delta_{jr}, 0); \begin{pmatrix} 1 & 0 & \frac{\rho}{\sqrt{2}} \\ 0 & 1 & \frac{\rho}{\sqrt{2}} \\ \frac{\rho}{\sqrt{2}}  & \frac{\rho}{\sqrt{2}} & 1 \end{pmatrix}) + \nonumber \\
& \Tilde{\Phi}_4((-\Delta_{jr}, -\infty, -\infty, -\infty), (-\Delta_{j(r-1)}, \Delta_{jr}, \Delta_m, 0); S_{5a}(\rho))] \nonumber \\
& = \sum_{r=1}^{l_{j}-1}[T_{1r}^{(2)} + T_{2r}^{(2)} - T_{3r}^{(2)} - T_{4r}^{(2)} + T_{5r}^{(2)}] \label{eq:trunc-deriv-01}
\end{align}

where, $S_{5a}(\rho)= \textrm{cov}(-L_{ij}, L_{i'j}, L_{i'm}, \frac{L_{i'm}- L_{im}}{\sqrt{2}}) =  \begin{pmatrix}
1 & 0 & 0 & \frac{\rho}{\sqrt{2}}\\
0 & 1 & \rho & \frac{\rho}{\sqrt{2}} \\
0 & \rho & 1 & \frac{1}{\sqrt{2}} \\
\frac{\rho}{\sqrt{2}} & \frac{\rho}{\sqrt{2}} & \frac{1}{\sqrt{2}} & 1
\end{pmatrix}$ and  \\
$S_{5b}(\rho)=  cov(L_{ij}, -L_{i'j}, - L_{i'm}, \frac{L_{i'm}- L_{im}}{\sqrt{2}}) = \begin{pmatrix}
1 & 0 & 0 & -\frac{\rho}{\sqrt{2}}\\
0 & 1 & \rho & -\frac{\rho}{\sqrt{2}} \\
0 & \rho & 1 & -\frac{1}{\sqrt{2}} \\
-\frac{\rho}{\sqrt{2}} & -\frac{\rho}{\sqrt{2}} & -\frac{1}{\sqrt{2}} & 1
\end{pmatrix}$. \\

\begin{align}
& \sum_{r=1}^{l_{j}-1}[T_{2r}^{(1)} + T_{4r}^{(2)}] \nonumber\\
& = \sum_{r=1}^{l_{j}-1}[\Phi_3((-\Delta_{jr}, \Delta_{jr}, 0); \begin{pmatrix} 1 & 0 & \frac{\rho}{\sqrt{2}} \\
0 & 1 & \frac{\rho}{\sqrt{2}} \\
\frac{\rho}{\sqrt{2}}  & \frac{\rho}{\sqrt{2}} & 1
\end{pmatrix}) - \Phi_3((-\Delta_{j(r+1)}, \Delta_{jr}, 0); \begin{pmatrix} 1 & 0 & \frac{\rho}{\sqrt{2}} \\
0 & 1 & \frac{\rho}{\sqrt{2}} \\
\frac{\rho}{\sqrt{2}}  & \frac{\rho}{\sqrt{2}} & 1
\end{pmatrix}) + \nonumber\\
& \Phi_3((-\Delta_{j(r-1)}, \Delta_{jr}, 0); \begin{pmatrix} 1 & 0 & \frac{\rho}{\sqrt{2}} \\ 0 & 1 & \frac{\rho}{\sqrt{2}} \\ \frac{\rho}{\sqrt{2}}  & \frac{\rho}{\sqrt{2}} & 1 \end{pmatrix}) - \Phi_3((-\Delta_{jr}, \Delta_{jr}, 0); \begin{pmatrix} 1 & 0 & \frac{\rho}{\sqrt{2}} \\
0 & 1 & \frac{\rho}{\sqrt{2}} \\
\frac{\rho}{\sqrt{2}}  & \frac{\rho}{\sqrt{2}} & 1
\end{pmatrix}) \nonumber\\
& = \sum_{r=1}^{l_{j}-1}[\Phi_3((-\Delta_{j(r-1)}, \Delta_{jr}, 0); \begin{pmatrix} 1 & 0 & \frac{\rho}{\sqrt{2}} \\ 0 & 1 & \frac{\rho}{\sqrt{2}} \\ \frac{\rho}{\sqrt{2}}  & \frac{\rho}{\sqrt{2}} & 1 \end{pmatrix}) - \Phi_3((-\Delta_{j(r+1)}, \Delta_{jr}, 0); \begin{pmatrix} 1 & 0 & \frac{\rho}{\sqrt{2}} \\
0 & 1 & \frac{\rho}{\sqrt{2}} \\
\frac{\rho}{\sqrt{2}}  & \frac{\rho}{\sqrt{2}} & 1
\end{pmatrix})] \nonumber\\
& = \sum_{r=1}^{l_{j}-1}[\Phi_2(-\Delta_{j(r-1)}, \Delta_{jr};0) - \Phi_3((-\Delta_{j(r-1)}, \Delta_{jr}, 0); \begin{pmatrix} 1 & 0 & -\frac{\rho}{\sqrt{2}} \\ 0 & 1 & -\frac{\rho}{\sqrt{2}} \\ -\frac{\rho}{\sqrt{2}}  &  -\frac{\rho}{\sqrt{2}} & 1 \end{pmatrix}) - \nonumber\\
& \Phi_2(-\Delta_{j(r-1)}, \Delta_{jr};0) + 
\Phi_3((-\Delta_{j(r+1)}, \Delta_{jr}, 0); \begin{pmatrix} 1 & 0 & -\frac{\rho}{\sqrt{2}} \\
0 & 1 & -\frac{\rho}{\sqrt{2}} \\
-\frac{\rho}{\sqrt{2}}  & -\frac{\rho}{\sqrt{2}} & 1
\end{pmatrix})] \nonumber\\
& = \sum_{r=1}^{l_{j}-1}[\Phi(-\Delta_{j(r-1)}) \Phi(\Delta_{jr}) - \Phi(-\Delta_{j(r+1)}) \Phi(\Delta_{jr})] - \nonumber\\
& \sum_{r=1}^{l_{j}-1}[\Phi_3((-\Delta_{j(r-1)}, \Delta_{jr}, 0); \begin{pmatrix} 1 & 0 & -\frac{\rho}{\sqrt{2}} \\
0 & 1 & -\frac{\rho}{\sqrt{2}} \\
-\frac{\rho}{\sqrt{2}}  & -\frac{\rho}{\sqrt{2}} & 1
\end{pmatrix})- \nonumber \\
& \Phi_3((-\Delta_{j(r+1)}, \Delta_{jr}, 0); \begin{pmatrix} 1 & 0 & -\frac{\rho}{\sqrt{2}} \\
0 & 1 & -\frac{\rho}{\sqrt{2}} \\
-\frac{\rho}{\sqrt{2}}  & -\frac{\rho}{\sqrt{2}} & 1
\end{pmatrix})] \nonumber\\
& = \Phi(\Delta_{j(l_j-1)}) - \sum_{r=1}^{l_{j}-1}[\Phi_3((-\Delta_{j(r-1)}, \Delta_{jr}, 0); \begin{pmatrix} 1 & 0 & -\frac{\rho}{\sqrt{2}} \\
0 & 1 & -\frac{\rho}{\sqrt{2}} \\
-\frac{\rho}{\sqrt{2}}  & -\frac{\rho}{\sqrt{2}} & 1
\end{pmatrix})- \nonumber\\
& \Phi_3((-\Delta_{j(r+1)}, \Delta_{jr}, 0); \begin{pmatrix} 1 & 0 & -\frac{\rho}{\sqrt{2}} \\
0 & 1 & -\frac{\rho}{\sqrt{2}} \\
-\frac{\rho}{\sqrt{2}}  & -\frac{\rho}{\sqrt{2}} & 1
\end{pmatrix})] \label{eq:trunc-deriv-1}
\end{align} 

\begin{align}
& \sum_{r=1}^{l_{j}-1}[T_{3r}^{(1)} + T_{5r}^{(2)}] \nonumber \\
& = \sum_{r=1}^{l_{j}-1}[\Phi_4((-\Delta_{jr}, \Delta_{jr}, \Delta_m, 0); S_{5a}(\rho)) - \Phi_4((-\Delta_{j(r+1)}, \Delta_{jr}, \Delta_m, 0); S_{5a}(\rho)) + \nonumber \\
& -\Phi_4((-\Delta_{jr}, \Delta_{jr}, \Delta_m, 0); S_{5a}(\rho)) + \Phi_4((-\Delta_{j(r-1)}, \Delta_{jr}, \Delta_m, 0); S_{5a}(\rho))] \nonumber\\
& = \sum_{r=1}^{l_{j}-1}[\Phi_4((-\Delta_{j(r-1)}, \Delta_{jr}, \Delta_m, 0); S_{5a}(\rho)) - \Phi_4((-\Delta_{j(r+1)}, \Delta_{jr}, \Delta_m, 0); S_{5a}(\rho))] \nonumber\\
& \sum_{r=1}^{l_{j}-1}[T_{2r}^{(2)} - T_{3r}^{(2)}] \nonumber\\
& = \sum_{r=1}^{l_{j}-1}[\Tilde{\Phi}_2((-\Delta_{jr}, -\infty), (-\Delta_{j(r-1)}, 0); \frac{\rho}{\sqrt{2}})] - \nonumber\\ & \sum_{r=1}^{l_{j}-1}[\Tilde{\Phi}_3((-\Delta_{jr}, -\infty, -\infty), (-\Delta_{j(r-1)}, \Delta_m, 0); \begin{pmatrix} 1 & 0 & \frac{\rho}{\sqrt{2}} \\
0 & 1 & \frac{1}{\sqrt{2}} \\
\frac{\rho}{\sqrt{2}}  & \frac{1}{\sqrt{2}} & 1
\end{pmatrix}]\nonumber \\
& = \sum_{r=1}^{l_{j}-1}[\Phi_2( (-\Delta_{j(r-1)}, 0); \frac{\rho}{\sqrt{2}}) - \Phi_2( (-\Delta_{jr}, 0); \frac{\rho}{\sqrt{2}})] - \nonumber\\ 
& \sum_{r=1}^{l_{j}-1}[\Phi_3( (-\Delta_{j(r-1)}, \Delta_m, 0); \begin{pmatrix} 1 & 0 & \frac{\rho}{\sqrt{2}} \\
0 & 1 & \frac{1}{\sqrt{2}} \\
\frac{\rho}{\sqrt{2}}  & \frac{1}{\sqrt{2}} & 1
\end{pmatrix}) - \Phi_3( (-\Delta_{jr}, \Delta_m, 0); \begin{pmatrix} 1 & 0 & \frac{\rho}{\sqrt{2}} \\
0 & 1 & \frac{1}{\sqrt{2}} \\
\frac{\rho}{\sqrt{2}}  & \frac{1}{\sqrt{2}} & 1
\end{pmatrix})] \nonumber\\
& = \frac{1}{2} -  \Phi_2(-\Delta_{j(l_j-1)},0; \frac{\rho}{\sqrt{2}}) - ( \Phi_2((\Delta_m, 0); \frac{1}{\sqrt{2}}) -  \Phi_3( (-\Delta_{j(l_j-1)}, \Delta_m, 0); \begin{pmatrix} 1 & 0 & \frac{\rho}{\sqrt{2}} \\
0 & 1 & \frac{1}{\sqrt{2}} \\
\frac{\rho}{\sqrt{2}}  & \frac{1}{\sqrt{2}} & 1
\end{pmatrix})) \nonumber\\
& =\Phi_2(\Delta_{j(l_j-1)},0; -\frac{\rho}{\sqrt{2}}) - \Phi_3( (\Delta_{j(l_j-1)}, \Delta_m, 0); \begin{pmatrix} 1 & 0 & \frac{-\rho}{\sqrt{2}} \\
0 & 1 & \frac{1}{\sqrt{2}} \\
\frac{-\rho}{\sqrt{2}}  & \frac{1}{\sqrt{2}} & 1
\end{pmatrix})) \nonumber\\
& = \Phi_3( (\Delta_{j(l_j-1)}, -\Delta_m, 0); \begin{pmatrix} 1 & 0 & \frac{-\rho}{\sqrt{2}} \\
0 & 1 & \frac{-1}{\sqrt{2}} \nonumber \\
\frac{-\rho}{\sqrt{2}}  & \frac{-1}{\sqrt{2}} & 1
\end{pmatrix})) \nonumber\\
\label{eq:trunc-deriv-2}
\end{align}

\begin{align}
& \sum_{r=1}^{l_{j}-1}[T_{1r}^{(1)} -  T_{1r}^{(2)}] \nonumber\\
& =\sum_{r=1}^{l_{j}-1} [\Tilde{\Phi}_2((\Delta_{jr},-\infty),(\Delta_{j(r+1)},-\Delta_m);-\rho) \Tilde{\Phi}_2((-\infty,-\infty),(\Delta_{jr},\Delta_{m});\rho) -\nonumber\\
&\Tilde{\Phi}_2((\Delta_{j(r-1)},-\infty),(\Delta_{jr},-\Delta_m);-\rho) \Tilde{\Phi}_2((-\infty,-\infty),(-\Delta_{jr},\Delta_{m});-\rho)]\nonumber \\
& = \sum_{r=1}^{l_{j}-1}[\Phi_2 (\Delta_{jr},\Delta_{m}; \rho)\{\Phi_2(\Delta_{jr},\Delta_m;\rho) + \Phi(\Delta_{j(r+1)}) - \Phi(\Delta_{jr}) - \Phi_2(\Delta_{j(r+1)},\Delta_m;\rho)\} - \nonumber\\
 & (\Phi(\Delta_m) - \Phi_2(\Delta_{jr},\Delta_m;\rho))\{\Phi_2(\Delta_{j(r-1)},\Delta_m;\rho) + \Phi(\Delta_{jr}) - \Phi(\Delta_{j(r-1)}) - \Phi_2(\Delta_{jr},\Delta_m;\rho)\}] \nonumber\\
 & = \sum_{r=1}^{l_{j}-1}[\Phi_2 (\Delta_{jr},\Delta_{m}; \rho)\{\Phi_2(\Delta_{jr},\Delta_m;\rho) + \Phi(\Delta_{j(r+1)}) - \Phi(\Delta_{jr}) - \nonumber\\
& \Phi_2(\Delta_{j(r+1)},\Delta_m;\rho) + \Phi_2(\Delta_{j(r-1)},\Delta_m;\rho) + \Phi(\Delta_{jr}) - \Phi(\Delta_{j(r-1)}) - \Phi_2(\Delta_{jr},\Delta_m;\rho) + \Phi(\Delta_m)\}- \nonumber\\
& \Phi(\Delta_m)\{\Phi_2(\Delta_{j(r-1)},\Delta_m;\rho) + \Phi(\Delta_{jr}) - \Phi(\Delta_{j(r-1)})\}]  \nonumber\\
& = \sum_{r=1}^{l_{j}-1}[\Phi_2 (\Delta_{jr},\Delta_{m}; \rho)\{\Phi_2(\Delta_{jr},\Delta_m;\rho) + \Phi(\Delta_{j(r+1)}) - \Phi(\Delta_{jr}) - \Phi_2(\Delta_{j(r+1)},\Delta_m;\rho)\} - \nonumber\\
 & (\Phi(\Delta_m) - \Phi_2(\Delta_{jr},\Delta_m;\rho))\{\Phi_2(\Delta_{j(r-1)},\Delta_m;\rho) + \Phi(\Delta_{jr}) - \Phi(\Delta_{j(r-1)}) - \Phi_2(\Delta_{jr},\Delta_m;\rho)\}] + \nonumber\\
& = \sum_{r=1}^{l_{j}-1}[\Phi_2 (\Delta_{jr},\Delta_{m}; \rho)\{\Phi(\Delta_{j(r+1)}) - \Phi(\Delta_{j(r-1)})\} + \nonumber\\
& \{\Phi_2 (\Delta_{jr},\Delta_{m}; \rho) \Phi_2(\Delta_{j(r-1)},\Delta_m;\rho) - \Phi_2 (\Delta_{j(r+1)},\Delta_{m}; \rho) \Phi_2(\Delta_{jr},\Delta_m;\rho)\}]  - \nonumber\\ 
& \Phi(\Delta_m)[\sum_{r=1}^{l_{j}-1}\{\Phi_2(\Delta_{j(r-1)},\Delta_m;\rho) - \Phi_2(\Delta_{jr},\Delta_m;\rho)\} + \sum_{r=1}^{l_{j}-1} \{\Phi(\Delta_{jr}) - \Phi(\Delta_{j(r-1)})\}] + \nonumber\\
& = \sum_{r=1}^{l_{j}-1}[\Phi_2 (\Delta_{jr},\Delta_{m}; \rho)\{\Phi(\Delta_{j(r+1)}) - \Phi(\Delta_{j(r-1)})\}] \nonumber\\
& - \Phi(\Delta_m) \Phi_2(\Delta_{j(l_j-1)},\Delta_m;\rho) +  2 \Phi(\Delta_m)\Phi_2(\Delta_{j(l_j-1)},\Delta_m;\rho) -  \Phi(\Delta_m) \Phi(\Delta_{j(l_j-1)}) \nonumber\\
& =\sum_{r=1}^{l_{j}-1}[\Phi_2 (\Delta_{jr},\Delta_{m}; \rho)\{\Phi(\Delta_{j(r+1)}) - \Phi(\Delta_{j(r-1)})\}] -  \Phi(\Delta_m) \Phi(\Delta_{j(l_j-1)}) \nonumber\\
&=\sum_{r=1}^{l_{j}-1}[\Phi_2 (\Delta_{jr},\Delta_{m}; \rho)\{\Phi(-\Delta_{j(r-1)}) - \Phi(\Delta_{j(r+1)})\}] -  \Phi(\Delta_m) \Phi(\Delta_{j(l_j-1)})\nonumber\\
& = \sum_{r=1}^{l_{j}-1}[\Phi_3( (-\Delta_{j(r-1)}, \Delta_{jr}, \Delta_{m}); \begin{pmatrix} 1 & 0 & 0 \\
0 & 1 & \rho \\
0  & \rho & 1
\end{pmatrix})- \Phi_3( (-\Delta_{j(r+1)}, \Delta_{jr}, \Delta_{m}); \begin{pmatrix} 1 & 0 & 0 \\
0 & 1 & \rho \\
0  & \rho & 1
\end{pmatrix})] - \nonumber\\
& \Phi(\Delta_m) \Phi(\Delta_{j(l_j-1)}) \nonumber\\
\label{eq:trunc-deriv-3}
\end{align}    

Using the fact that,\\  $\Phi_3( (-\Delta_{jr}, \Delta_{js}, \Delta_{m}); \begin{pmatrix} 1 & 0 & 0 \\
0 & 1 & \rho \\
0  & \rho & 1
\end{pmatrix}) = \Phi_4((-\Delta_{jr}, \Delta_{js}, \Delta_m, 0); \begin{pmatrix}
1 & 0 & 0 & \frac{\rho}{\sqrt{2}}\\
0 & 1 & \rho & \frac{\rho}{\sqrt{2}} \\
0 & \rho & 1 & \frac{1}{\sqrt{2}} \\
\frac{\rho}{\sqrt{2}} & \frac{\rho}{\sqrt{2}} & \frac{1}{\sqrt{2}} & 1
\end{pmatrix}) + \\ \Phi_4((-\Delta_{jr}, \Delta_{js}, \Delta_m, 0); \begin{pmatrix}
1 & 0 & 0 & -\frac{\rho}{\sqrt{2}}\\
0 & 1 & \rho & -\frac{\rho}{\sqrt{2}} \\
0 & \rho & 1 & -\frac{1}{\sqrt{2}} \\
-\frac{\rho}{\sqrt{2}} & -\frac{\rho}{\sqrt{2}} & -\frac{1}{\sqrt{2}} & 1
\end{pmatrix})$ and combining equations \eqref{eq:trunc-deriv-2} and \eqref{eq:trunc-deriv-3} we get - 

\begin{equation}
\begin{split}
\sum_{r=1}^{l_{j}-1}[T_{1r}^{(1)} -  T_{1r}^{(2)} - T_{3r}^{(1)} - T_{5r}^{(2)}] & = \sum_{r=1}^{l_{j}-1}[\Phi_4((-\Delta_{j(r-1)}, \Delta_{jr}, \Delta_m, 0); \begin{pmatrix}
1 & 0 & 0 & -\frac{\rho}{\sqrt{2}}\\
0 & 1 & \rho & -\frac{\rho}{\sqrt{2}} \\
0 & \rho & 1 & -\frac{1}{\sqrt{2}} \\
-\frac{\rho}{\sqrt{2}} & -\frac{\rho}{\sqrt{2}} & -\frac{1}{\sqrt{2}} & 1
\end{pmatrix}) -\\
& \Phi_4((-\Delta_{j(r+1)}, \Delta_{jr}, \Delta_m, 0); \begin{pmatrix}
1 & 0 & 0 & -\frac{\rho}{\sqrt{2}}\\
0 & 1 & \rho & -\frac{\rho}{\sqrt{2}} \\
0 & \rho & 1 & -\frac{1}{\sqrt{2}} \\
-\frac{\rho}{\sqrt{2}} & -\frac{\rho}{\sqrt{2}} & -\frac{1}{\sqrt{2}} & 1
\end{pmatrix})] - \Phi(\Delta_m) \Phi(\Delta_{j(l_j-1)})
\end{split}
\label{eq:trunc-deriv-4}
\end{equation}

Next, using similar result that, \\ $\Phi_3((-\Delta_{jr}, \Delta_{js}, 0); \begin{pmatrix} 1 & 0 & -\frac{\rho}{\sqrt{2}} \\
0 & 1 & -\frac{\rho}{\sqrt{2}} \\
-\frac{\rho}{\sqrt{2}}  & -\frac{\rho}{\sqrt{2}} & 1
\end{pmatrix} = \\ \Phi_4((-\Delta_{jr}, \Delta_{js}, 0, \Delta_m); \begin{pmatrix} 1 & 0 & -\frac{\rho}{\sqrt{2}} & 0 \\
0 & 1 & -\frac{\rho}{\sqrt{2}} & \rho \\
-\frac{\rho}{\sqrt{2}}  & -\frac{\rho}{\sqrt{2}} & 1 & -\frac{1}{\sqrt{2}} \\
0 & \rho & -\frac{1}{\sqrt{2}} & 1
\end{pmatrix}) + \\ \Phi_4((-\Delta_{jr}, \Delta_{js}, 0, -\Delta_m); \begin{pmatrix} 1 & 0 & -\frac{\rho}{\sqrt{2}} & 0 \\
0 & 1 & -\frac{\rho}{\sqrt{2}} & -\rho \\
-\frac{\rho}{\sqrt{2}}  & -\frac{\rho}{\sqrt{2}} & 1 & \frac{1}{\sqrt{2}} \\
0 & -\rho & \frac{1}{\sqrt{2}} & 1
\end{pmatrix})$, and combining equations \eqref{eq:trunc-deriv-1} and \eqref{eq:trunc-deriv-4} we get - 

\begin{equation}
\begin{split}
\sum_{r=1}^{l_{j}-1}[T_{1r}^{(1)} -  T_{1r}^{(2)} - T_{3r}^{(1)} - T_{5r}^{(2)} +T_{2r}^{(1)} + T_{4r}^{(2)}] & = \Phi(\Delta_{j(l_j-1)})(1-\Phi(\Delta_m)) -\\
& \sum_{r=1}^{l_{j}-1}[\Phi_4((-\Delta_{j(r-1)}, \Delta_{jr}, 0, -\Delta_m); \begin{pmatrix} 1 & 0 & -\frac{\rho}{\sqrt{2}} & 0 \\
0 & 1 & -\frac{\rho}{\sqrt{2}} & -\rho \\
-\frac{\rho}{\sqrt{2}}  & -\frac{\rho}{\sqrt{2}} & 1 & \frac{1}{\sqrt{2}} \\
0 & -\rho & \frac{1}{\sqrt{2}} & 1
\end{pmatrix})-\\
& \Phi_4((-\Delta_{j(r+1)}, \Delta_{jr}, 0, -\Delta_m); \begin{pmatrix} 1 & 0 & -\frac{\rho}{\sqrt{2}} & 0 \\
0 & 1 & -\frac{\rho}{\sqrt{2}} & -\rho \\
-\frac{\rho}{\sqrt{2}}  & -\frac{\rho}{\sqrt{2}} & 1 & \frac{1}{\sqrt{2}} \\
0 & -\rho & \frac{1}{\sqrt{2}} & 1
\end{pmatrix})]
\end{split}
\label{eq:trunc-deriv-5}
\end{equation}

Combining equations \eqref{eq:trunc-deriv-00} and \eqref{eq:trunc-deriv-01} and putting the values from equation \eqref{eq:trunc-deriv-5}, we get -

\begin{align}
& F_{to}(\rho; \Delta_j, \Delta_m) \nonumber\\
& = 2 \sum_{r=1}^{l_{j}-1}[T_{1r}^{(1)} -  T_{1r}^{(2)} - T_{3r}^{(1)} - T_{5r}^{(2)} +T_{2r}^{(1)} + T_{4r}^{(2)} - T_{2r}^{(2)} + T_{3r}^{(2)}] \nonumber \\
& = 2 \Phi(\Delta_{j(l_j-1)})(1-\Phi(\Delta_m)) - 2 \Phi_3( (\Delta_{j(l_j-1)}, -\Delta_m, 0); \begin{pmatrix} 1 & 0 & \frac{-\rho}{\sqrt{2}} \\
0 & 1 & \frac{-1}{\sqrt{2}} \\
\frac{-\rho}{\sqrt{2}}  & \frac{-1}{\sqrt{2}} & 1
\end{pmatrix})) - \nonumber \\
&  2 \sum_{r=1}^{l_{j}-1}[\Phi_4((-\Delta_{j(r-1)}, \Delta_{jr}, 0, -\Delta_m); \begin{pmatrix} 1 & 0 & -\frac{\rho}{\sqrt{2}} & 0 \\
0 & 1 & -\frac{\rho}{\sqrt{2}} & -\rho \\
-\frac{\rho}{\sqrt{2}}  & -\frac{\rho}{\sqrt{2}} & 1 & \frac{1}{\sqrt{2}} \\
0 & -\rho & \frac{1}{\sqrt{2}} & 1
\end{pmatrix})- \nonumber \\
& \Phi_4((-\Delta_{j(r+1)}, \Delta_{jr}, 0, -\Delta_m); \begin{pmatrix} 1 & 0 & -\frac{\rho}{\sqrt{2}} & 0 \\
0 & 1 & -\frac{\rho}{\sqrt{2}} & -\rho \\
-\frac{\rho}{\sqrt{2}}  & -\frac{\rho}{\sqrt{2}} & 1 & \frac{1}{\sqrt{2}} \\
0 & -\rho & \frac{1}{\sqrt{2}} & 1
\end{pmatrix})] \nonumber \\
& = 2 \Phi(\Delta_{j(l_j-1)})(1-\Phi(\Delta_m)) - 2 \Phi_3( (\Delta_{j(l_j-1)}, -\Delta_m, 0); \begin{pmatrix} 1 & 0 & \frac{-\rho}{\sqrt{2}} \\
0 & 1 & \frac{-1}{\sqrt{2}} \\
\frac{-\rho}{\sqrt{2}}  & \frac{-1}{\sqrt{2}} & 1
\end{pmatrix})) - \nonumber \\
&  2 \sum_{r=1}^{l_{j}-1}[\Phi_4((-\Delta_{j(r-1)}, \Delta_{jr}, -\Delta_m, 0); \begin{pmatrix} 1 & 0 &  0 & -\frac{\rho}{\sqrt{2}} \\
0 & 1 & -\rho & -\frac{\rho}{\sqrt{2}} \\
0  & -\rho & 1 & \frac{1}{\sqrt{2}} \\
-\frac{\rho}{\sqrt{2}} & -\frac{\rho}{\sqrt{2}} & \frac{1}{\sqrt{2}} & 1
\end{pmatrix})- \nonumber \\
& \Phi_4((-\Delta_{j(r+1)}, \Delta_{jr}, -\Delta_m, 0); \begin{pmatrix} 1 & 0 &  0 & -\frac{\rho}{\sqrt{2}} \\
0 & 1 & -\rho & -\frac{\rho}{\sqrt{2}} \\
0  & -\rho & 1 & \frac{1}{\sqrt{2}} \\
-\frac{\rho}{\sqrt{2}} & -\frac{\rho}{\sqrt{2}} & \frac{1}{\sqrt{2}} & 1
\end{pmatrix})] \nonumber \\
& = 2 \Phi_3( (\Delta_{j(l_j-1)}, -\Delta_m, 0); \begin{pmatrix} 1 & 0 & \frac{\rho}{\sqrt{2}} \\
0 & 1 & \frac{1}{\sqrt{2}} \\
\frac{\rho}{\sqrt{2}}  & \frac{1}{\sqrt{2}} & 1
\end{pmatrix})) - \nonumber \\
&  2 \sum_{r=1}^{l_{j}-1}[\Phi_4((\Delta_{j(r+1)}, \Delta_{jr}, -\Delta_m, 0); \begin{pmatrix} 1 & 0 &  0 & \frac{\rho}{\sqrt{2}} \\
0 & 1 & -\rho & -\frac{\rho}{\sqrt{2}} \\
0  & -\rho & 1 & \frac{1}{\sqrt{2}} \\
\frac{\rho}{\sqrt{2}} & -\frac{\rho}{\sqrt{2}} & \frac{1}{\sqrt{2}} & 1
\end{pmatrix})- \nonumber \\
& \Phi_4((\Delta_{j(r-1)}, \Delta_{jr}, -\Delta_m, 0); \begin{pmatrix} 1 & 0 &  0 & \frac{\rho}{\sqrt{2}} \\
0 & 1 & -\rho & -\frac{\rho}{\sqrt{2}} \\
0  & -\rho & 1 & \frac{1}{\sqrt{2}} \\
\frac{\rho}{\sqrt{2}} & -\frac{\rho}{\sqrt{2}} & \frac{1}{\sqrt{2}} & 1
\end{pmatrix})] \label{eq:bridge-to}
\end{align}

\textbf{Proof of Theorem \ref{thm:bridge-inv}}:

We'll show that the bridging functions are strictly increasing in $\rho \in (-1,1)$ and hence invertible. 

\textbf{Case 1: Ordinal-Ordinal}

Let's recall the bridging function $F_{\rm oo}(\rho; \Delta_j, \Delta_k) = 2 \sum_{r=1}^{l_j-1}\sum_{s=1}^{l_k-1}[\Phi_2(\Delta_{jr},\Delta_{ks}; \rho)\\\{ \Phi_2(\Delta_{j(r+1)},\Delta_{k(s+1)};\rho) - \Phi_2(\Delta_{j(r+1)},\Delta_{k(s-1)}; \rho)\}] -  2 \sum_{r=1}^{l_j-1}\Phi({\Delta_{jr}})\Phi_2(\Delta_{j(r+1)},\Delta_{k(l_k-1)}; \rho)$. For the proof, we use the fact that $\partial_{r,s}^{(2)} = \frac{\partial{\Phi_2(\Delta_{jr},\Delta_{ks},\rho)}}{{\partial \rho}} > 0$ \citep[Lemma A1]{yoon2020sparse}. Also note that, since, $\Phi_2(\Delta_{jlj}, \Delta_{ks}, \rho) = \Phi(\Delta_{ks})$ and $\Phi_2(\Delta_{jr}, \Delta_{kl_k}, \rho) = \Phi(\Delta_{jr})$, we have $\partial_{r,l_k}^{2} = \partial_{l_j,s}^{2} =0$. 

\begin{equation}
\begin{aligned}
& \frac{\partial{F_{\rm oo}(\rho; \Delta_j, \Delta_k)}}{\partial \rho} \\
& =  2 \sum_{r=1}^{l_j-1}\sum_{s=1}^{l_k-1}[\Phi_2(\Delta_{jr},\Delta_{ks}; \rho)\{\partial_{(r+1),(s+1)}^{(2)} - \partial_{(r+1),(s-1)}^{(2)}\} +  \partial_{r,s}^{(2)} \{ \Phi_2(\Delta_{j(r+1)},\Delta_{k(s+1)};\rho) - \Phi_2(\Delta_{j(r+1)},\Delta_{k(s-1)}; \rho)\}]  - \\
& 2 \sum_{r=1}^{l_j-1}\Phi({\Delta_{jr}})\partial_{(r+1),(l_k-1)}^{(2)} \\
& = 2 \sum_{r=1}^{l_j-1}\sum_{s=1}^{l_k-2}[\partial_{r,s}^{(2)}\{\Phi_2(\Delta_{j(r+1)},\Delta_{k(s+1)};\rho) + \Phi_2(\Delta_{j(r-1)},\Delta_{k(s-1)};\rho) -\\
& \Phi_2(\Delta_{j(r+1)},\Delta_{k(s-1)};\rho) - \Phi_2(\Delta_{j(r-1)},\Delta_{k(s+1)};\rho)\}] +\\
& 2 \sum_{r=1}^{l_j-1} [\partial_{r,(l_k-1)}^{(2)}\{\Phi(\Delta_{j(r+1)}) - \Phi_2(\Delta_{j(r+1)}, \Delta_{k(l_k-2)}; \rho) + \Phi_2(\Delta_{j(r-1)}, \Delta_{k(l_k-2)}; \rho)\}] - \\ 
& 2 \sum_{r=1}^{l_j-1}\Phi({\Delta_{jr}})\partial_{(r+1),(l_k-1)}^{(2)} \\
& = 2 \sum_{r=1}^{l_j-1}\sum_{s=1}^{l_k-2}[\partial_{r,s}^{(2)}[\Tilde{\Phi_2}((\Delta_{j(r-1)}, \Delta_{k(s-1)}), (\Delta_{j(r+1)},\Delta_{k(s+1)}); \rho)] +\\
& 2 \sum_{r=1}^{l_j-1} [\partial_{r,(l_k-1)}^{(2)}\{\Phi(\Delta_{j(r+1)}) - \Phi_2(\Delta_{j(r+1)}, \Delta_{k(l_k-2)}; \rho) + \Phi_2(\Delta_{j(r-1)}, \Delta_{k(l_k-2)}; \rho) - \Phi(\Delta_{j(r-1)})\}] - \\
& \Phi(\Delta_{j0}) \partial_{1,(l_k-1)}^{(2)}\\
& = 2 \sum_{r=1}^{l_j-1}\sum_{s=1}^{l_k-2}[\partial_{r,s}^{(2)}[\Tilde{\Phi_2}((\Delta_{j(r-1)}, \Delta_{k(s-1)}), (\Delta_{j(r+1)},\Delta_{k(s+1)}); \rho)] + \\
& 2 \sum_{r=1}^{l_j-1} [\partial_{r,(l_k-1)}^{(2)} \Tilde{\Phi_2}((\Delta_{j(r-1)}, \Delta_{k(l_k-2)}), (\Delta_{j(r+1)},\infty)); \rho)] \\
& > 0 
\end{aligned}    
\end{equation}

\textbf{Case 2: Ordinal-Binary}

Special case of the last proof.

\textbf{Case 3: Ordinal-Truncated}

Let's recall the bridging function (\eqref{eq:bridge-to})- \\ \begin{align*}
F_{to}(\rho; \Delta_j, \Delta_m) & =2 \Phi_3( (\Delta_{j(l_j-1)}, -\Delta_m, 0); \begin{pmatrix} 1 & 0 & \frac{\rho}{\sqrt{2}} \\
0 & 1 & \frac{1}{\sqrt{2}} \\
\frac{\rho}{\sqrt{2}}  & \frac{1}{\sqrt{2}} & 1
\end{pmatrix})) - \\ 
& 2 \sum_{r=1}^{l_{j}-1}[\Phi_4((\Delta_{j(r+1)}, \Delta_{jr}, -\Delta_m, 0); \begin{pmatrix} 1 & 0 &  0 & \frac{\rho}{\sqrt{2}} \\
0 & 1 & -\rho & -\frac{\rho}{\sqrt{2}} \\
0  & -\rho & 1 & \frac{1}{\sqrt{2}} \\
\frac{\rho}{\sqrt{2}} & -\frac{\rho}{\sqrt{2}} & \frac{1}{\sqrt{2}} & 1
\end{pmatrix})-
\\
& \Phi_4((\Delta_{j(r-1)}, \Delta_{jr}, -\Delta_m, 0); \begin{pmatrix} 1 & 0 &  0 & \frac{\rho}{\sqrt{2}} \\
0 & 1 & -\rho & -\frac{\rho}{\sqrt{2}} \\
0  & -\rho & 1 & \frac{1}{\sqrt{2}} \\
\frac{\rho}{\sqrt{2}} & -\frac{\rho}{\sqrt{2}} & \frac{1}{\sqrt{2}} & 1
\end{pmatrix})] \\
& = T_1   - \sum_{r=1}^{l_{j}-1}[T_{2r} - T_{3r}] 
\end{align*}

To calculate the derivative of $F_{to}(\rho; \Delta_j, \Delta_m)$, we use the following result from \cite[Lemma A1]{yoon2020sparse} - 
\begin{lemma}
For any constants $a_1,\cdots, a_d$, let $\Phi_d(a_1,a_2, \cdots, a_d, S(\rho))$ denote the cdf of a d-variate standard normal distribution with parametrized covariance matrix $S(\rho)= ((S_{ij}(\rho))$. Then, there exist $h_{ij}(t)>0$ for all $t\in (-1,1)$ such that - 
$$\frac{\partial{\Phi_d(a_1,a_2, \cdots, a_d)}}{\partial{\rho}} = \sum_{1\leq i < j \leq d} h_{ij}(t) \frac{\partial{S_{ij}(\rho)}}{\partial{\rho}}.$$
Here, $h_{12}(t) = \frac{\partial{\Phi_d(a_1,a_2, \cdots, a_d)}}{\partial{\rho_{12}(t)}} = \int_{-\infty}^{-a_3} \int_{-\infty}^{-a_4} \cdots \int_{-\infty}^{-a_d} \phi_d(a_1, a_2, x_3, \cdots, x_d; S(\rho)) dx_3 dx_4 \cdots dx_d$, where  $\phi_d(\cdot)$ denotes the density of the multivariate normal distribution. Other $h_{ij}(t)$'s are derived analogously. 
\end{lemma}

Using the above result, the derivative look like this - 

\begin{align*}
& \frac{\partial{F_{to}(\rho; \Delta_j, \Delta_m)}}{\rho} \\
& = \frac{\partial{T_1}}{\rho}  - \sum_{r=1}^{l_{j}-1}[ \frac{\partial{T_{2r}}}{\rho} + \frac{\partial{T_{3r}}}{\rho}] \\
& =2 h_{13}^{(1)}. (\frac{1}{\sqrt{2}}) - \\
& 2 \sum_{r=1}^{l_{j}-1}[h_{14}^{(2r)}. (\frac{1}{\sqrt{2}}) + h_{23}^{(2r)}. (-1) +  h_{24}^{(2r)}. (-\frac{1}{\sqrt{2}}) - h_{14}^{(3r)}. (\frac{1}{\sqrt{2}}) - h_{23}^{(3r)}. (-1) - h_{24}^{(3r)}. (-\frac{1}{\sqrt{2}})]\\
& = 2 \sum_{r=1}^{l_{j}-1}[h_{23}^{(2r)}- h_{23}^{(3r)}] + 2 \frac{1}{\sqrt{2}} \sum_{r=1}^{l_{j}-1}[h_{24}^{(2r)}- h_{24}^{(3r)}] + 2 \frac{1}{\sqrt{2}}[h_{13}^{(1)} -\sum_{r=1}^{l_{j}-1}\{h_{14}^{(2r)} - h_{14}^{(3r)}\}]\\
& = 2 \sum_{r=1}^{l_{j}-1}[\int_{-\infty}^{0} \int_{-\infty}^{\Delta_{j(r+1)}} \phi_4(x_1, \Delta_{jr},-\Delta_{m}, 0; S_5(\rho)) dx_1 dx_4 -\\
& \int_{-\infty}^{0} \int_{-\infty}^{\Delta_{j(r-1)}} \phi_4(x_1, \Delta_{jr},-\Delta_{m}, 0; S_5(\rho))dx_1 dx_4] + \\
& \sqrt{2} \sum_{r=1}^{l_{j}-1}[\int_{-\infty}^{-\Delta_m} \int_{-\infty}^{\Delta_{j(r+1)}} \phi_4(x_1, \Delta_{jr},-\Delta_{m}, 0; S_5(\rho)) dx_1 dx_3 - \\
& \int_{-\infty}^{-\Delta_m} \int_{-\infty}^{\Delta_{j(r-1)}} \phi_4(x_1, \Delta_{jr},x_3, 0; S_5(\rho)) dx_1 dx_3] + \\ 
& \sqrt{2}[\int_{-\infty}^{-\Delta_m} \phi_3(\Delta_{j(l_j-1)}, x_3, 0; S_{3a}(\rho))dx_3 - \\
& \sum_{r=1}^{l_{j}-1}\{\int_{-\infty}^{-\Delta_m} \int_{-\infty}^{\Delta_{jr}} \phi_4(\Delta_{j(r+1)}, x_2,x_3, 0; S_5(\rho)) dx_2 dx_3 - \int_{-\infty}^{-\Delta_m} \int_{-\infty}^{\Delta_{jr}} \phi_4(\Delta_{j(r-1)}, x_2,x_3, 0; S_5(\rho)) dx_2 dx_3\}]\\
& = 2 \sum_{r=1}^{l_{j}-1}[\int_{-\infty}^{0} \int_{\Delta_{j(r-1)}}^{\Delta_{j(r+1)}} \phi_4(x_1, \Delta_{jr},-\Delta_{m}, 0; S_5(\rho)) dx_1 dx_4 + \\
& \sqrt{2} \sum_{r=1}^{l_{j}-1}[\int_{-\infty}^{-\Delta_m} \int_{\Delta{j(r-1)}}^{\Delta_{j(r+1)}} \phi_4(x_1, \Delta_{jr},-\Delta_{m}, 0; S_5(\rho)) dx_1 dx_3 + \\
& \sqrt{2}[\int_{-\infty}^{-\Delta_m} \phi_3(\Delta_{j(l_j-1)}, x_3, 0; S_{3a}(\rho))dx_3 - \\
& \{\sum_{r=1}^{l_{j}-1}\int_{-\infty}^{-\Delta_m} \int_{-\infty}^{\Delta_{jr}} \phi_4(\Delta_{j(r+1)}, x_2,x_3, 0; S_5(\rho)) dx_2 dx_3 - \\
& \sum_{r=3}^{l_{j}-1} \int_{-\infty}^{-\Delta_m} \int_{-\infty}^{\Delta_{j(r-2)}} \phi_4(\Delta_{j(r-1)}, x_2,x_3, 0; S_5(\rho)) dx_2 dx_3\} - \\
& \sum_{r=3}^{l_{j}-1} \{\int_{-\infty}^{-\Delta_m} \int_{\Delta_j(r-2)}^{\Delta_{jr}} \phi_4(\Delta_{j(r-1)}, x_2,x_3, 0; S_5(\rho)) dx_2 dx_3\} -  \\
& \int_{-\infty}^{-\Delta_m} \int_{-\infty}^{\Delta_{j2}} \phi_4(\Delta_{j(r-1)}, x_2,x_3, 0; S_5(\rho)) dx_2 dx_3]\\
& > \sqrt{2}[\int_{-\infty}^{-\Delta_m} \phi_3(\Delta_{j(l_j-1)}, x_3, 0; S_{3a}(\rho))dx_3 - \int_{-\infty}^{-\Delta_m} \int_{-\infty}^{\Delta_{j(l_j-2)}} \phi_4(\Delta_{j(l_j-1)}, x_2,x_3, 0; S_5(\rho)) dx_2 dx_3] + \\
& \sum_{r=3}^{l_{j}-1} \{\int_{-\infty}^{-\Delta_m} \int_{\Delta_j(r-2)}^{\Delta_{jr}} \phi_4(\Delta_{j(r-1)}, x_2,x_3, 0; S_5(\rho)) dx_2 dx_3\} +  \int_{-\infty}^{-\Delta_m} \int_{-\infty}^{\Delta_{j2}} \phi_4(\Delta_{j(r-1)}, x_2,x_3, 0; S_5(\rho)) dx_2 dx_3\\
& > \sqrt{2}[\int_{-\infty}^{-\Delta_m}\int_{-\infty}^{\infty} \phi_4(\Delta_{j(l_j-1)}, x_2, x_3, 0; S_5(\rho))dx_2 dx_3 - \int_{-\infty}^{-\Delta_m} \int_{-\infty}^{\Delta_{j(l_j-2)}} \phi_4(\Delta_{j(l_j-1)}, x_2,x_3, 0; S_5(\rho)) dx_2 dx_3] \\
& = \sqrt{2}[\int_{-\infty}^{-\Delta_m}\int_{\Delta_{j(l_j-2)}}^{\infty} \phi_4(\Delta_{j(l_j-1)}, x_2, x_3, 0; S_5(\rho))dx_2 dx_3] \\
& > 0 
\end{align*}

\textbf{Proof of Theorem \ref{thm: asymp}}:

First, let's familiarize ourselves with some notations.  For a $p \times p$ correlation matrix $A$, we can get singular-value decomposition of $A$ as $A= Q \Lambda Q^T$, where $Q$ is an orthonormal matrix and $\Lambda= diag(\lambda_1,\cdots, \lambda_p)$ are the eigenvalues of $A$. Let's define $log A = Q log \Lambda Q^T$, where $log \Lambda = diag(log \lambda_1,\cdots, log \lambda_p)$. 

First we need to state the results for the asymptotic variance of Kendall's Tau as calculated in \cite{hoeffding1992class} and \cite{el2003spearman}. Using the results of U-statistics asymptotics, we state the results in the Lemma \ref{lem: kasymp} below. 

\begin{lemma}
Let $K_n$ be the Kendall's Tau matrix estimated from the data, then $\sqrt{n}(vecl(K_n) - vecl(K))$ is asymptotically normal with mean $0$ and variance-covariance matrix $V_K$, where, $K_{ij} = E(sgn((X_{i1}- X_{i2})(X_{j1} - X_{j2}))$ and $$V_{K_{(ij),(kl)}} = 4*(E(sgn((X_{i1}-X_{i2})(X_{j1}-X_{j2})(X_{k2}-X_{k3})(X_{l2}-X_{l3})))- (vecl(K)vecl(K)^T)_{(ij),(kl)})$$
$\{(ij),(kl)\}$ denotes the entries corresponding to the covariance of Kendall's tau between $(ij)$ and $(kl)$-th pair of variables. 
\label{lem: kasymp}
\end{lemma}

Now, we can rewrite the expression $E(sgn((X_{i1-X_i2})(X_{j1}-X_{j2})(X_{k2}-X_{k3})(X_{l2}-X_{l3})))$ as follows - 
\begin{equation}
\resizebox{\textwidth}{!}{$\begin{aligned}
& E(sgn((X_{i1-X_i2})(X_{j1}-X_{j2})(X_{k2}-X_{k3})(X_{l2}-X_{l3}))) \\
& = E(E(sgn((X_{i1}-X_{i2})(X_{j1}-X_{j2})(X_{k2}-X_{k3})(X_{l2}-X_{l3}))|(X_{i2},X_{j2},X_{k2},X_{l2})))\\
& = E(E(sgn((X_{i1}-X_{i2})(X_{j1}-X_{j2}))|(X_{i2},X_{j2},X_{k2},X_{l2}))*E(sgn((X_{k2}-X_{k3})(X_{l2}-X_{l3}))|(X_{i2},X_{j2},X_{k2},X_{l2})))\\
& = E(E(sgn((X_{i1}-X_{i2})(X_{j1}-X_{j2}))|(X_{i2},X_{j2}))*E(sgn((X_{k2}-X_{k3})(X_{l2}-X_{l3}))|(X_{k2},X_{l2})))\\
& = E(H_{ij}(X_{i2},X_{j2}) H_{kl}(X_{k2},X_{l2}))
\end{aligned}$}\label{eq: k-asymp}
\end{equation}
,where, $H_{ij}(x,y)= E(sgn((X_i-x)(X_j-y))$. We can estimate $H_{ij}(x,y)$ from sample as - $\hat{H}_{ij}(x,y)= \frac{1}{n}\sum_{m=1}^{n}(sgn((X_{im}-x)(X_{jm}-y))$

Hence, the quantity in \eqref{eq: k-asymp} can be estimated as  - $$
\frac{1}{n}\sum_{m=1}^{n}\hat{H}_{ij}(X_{im},X_{jm})\hat{H}_{kl}(X_{km},X_{lm}) = \frac{1}{n}\sum_{m=1}^{n} \hat{H}_m \hat{H}_m^T $$.

Here, $\hat{H}_m= (\hat{H}_{12}(X_{1m},X_{2m}), \cdots, \hat{H}_{(p-1),p}(X_{(p-1)m},X_{pm}))$ is a $\frac{p*(p-1)}{2}$ dimensional vector. 

\cite{perreault2022efficient} provides an efficient way of calculating $\hat{H}_m$ in $O(nlogn)$ FLOPs for continuous data. We extend this approach to account for mixed-type data and presence of ties. This approach is a significant improvement over calculating the quantity in \eqref{eq: k-asymp} blatantly which would have required $O(n^4)$ FLOPs. Hence, we provide a novel efficient way of calculating asymptotic variance of Kendall's Tau which would have been infeasible for even moderate $n$. 

Now we want to derive the asymptotic normality of $vecl(\Sigma_n)$ and $vec(\beta)$ using Delta method and the following result. 

As shown in \cite{archakov2018new}, a correlation matrix $A$ can be parametrized by $vecl(logA)$. There exists a bijective map $\gamma: \mathbb{C}_p \longrightarrow \mathbb{R}^{p(p-1)/2}$ which is defined by $\gamma(A)= vecl(logA)$, where $\mathbb{C}_p$ denotes the set of $p \times p$ correlation matrices. As described in \cite{tracy1988patterned}, a general technique of defining derivatives with respect to a structured matrix (such as a correlation matrix) is to first define a map from the matrix to the independent elements of the matrix and then extend the function under investigation to the set of general matrices. For example, let's take a function $h(A)$ of a correlation matrix $A$, then we will define the derivative as - 

$$ \frac{dvec(h(A))}{dvecl(\gamma(A))} = \frac{dvec(h(A))}{dvec(A)} \frac{dvec(A)}{dvecl(A)}\frac{dvecl(A)}{dvecl(\gamma(A))} $$. 

where, the first derivative $\frac{dvec(h(A))}{dvec(A)}$ is defined assuming $h$ is a general map defined on unstructured matrices. We can use this result and chain rule to derive the following - 
\begin{equation}
    \begin{aligned}
        \frac{dvecl(\Sigma)}{dvecl(\gamma(K))} &= \frac{dvecl(\Sigma)}{dvecl(K)} \frac{dvecl(K)}{dvecl(\gamma(K))} = \mathbb{D}_g{\Gamma} \\
        \frac{dvec(\beta)}{dvecl(\gamma(K))} &= \frac{dvec(\beta)}{dvec(\Sigma)} \frac{dvec(\Sigma)}{dvecl(\Sigma)} \frac{dvecl(\Sigma)}{dvecl(\gamma(K))}  = \mathbb{D}_\beta H_p \mathbb{D}_g{\Gamma} 
    \end{aligned}
    \label{eqn: asymp-sig-beta}
\end{equation}

Here, $H_p$ denotes duplication matrix of order $p$ which transforms $vecl(A)$ to $vec(A)$ for any matrix $A$, $\mathbb{D}_g = diag(g'(K))$, where $g'$ is the first order derivative of the indidividual bridging functions and then we apply it to $K$, $\mathbb{D}_\beta = ((\Sigma_{22}^{-1}, - (\Sigma_{21} \otimes I_p)(\Sigma_{22}^{-1} \otimes \Sigma_{22}^{-1})) \Tilde{E}$, where $\Tilde{E}$ transforms $vec(\Sigma)$ to $(vec(\Sigma_{21}),vec(\Sigma_{22}))$. $\Gamma$ is calculated in \cite{archakov2018new}. 

Now, to use the results in \eqref{eqn: asymp-sig-beta}, we first derive the asymptotic normality of $\sqrt{n}(vecl(\gamma(K_n) - \gamma(K))$ using results in \cite{archakov2018new} and calculate the asymptotic covariance matrix as $V_\gamma$. Then, under the regularity assumptions, we apply delta method to get asymptotic covariance matrix of $\sqrt{n}(vecl(\hat{\Sigma}_n- vecl(\Sigma))$ as $V_\Sigma = (\mathbb{D}_g{\Gamma})^T V_{\gamma} \mathbb{D}_g{\Gamma}$ and asymptotic covariance matrix of $\sqrt{n}(\hat{\beta_n}-\beta)$ as $V_\beta = 
(\mathbb{D}_\beta H_p \mathbb{D}_g{\Gamma} )^T V_{\gamma} \mathbb{D}_\beta H_p \mathbb{D}_g{\Gamma}$. Combining the previous derivations, the steps to get the asymptotic variance can be summarized as follows - 

\begin{enumerate}
    \item  Calculate variance of lower-triangular part of Kendall's Tau matrix - $V_{\tau}$ of dimension $ (\frac{p(p-1)}{2} \times \frac{p(p-1)}{2})$. Complexity: $O(n\log n p^4)$. But all the $O(p^4)$ calculcations can be done in parallel theoretically, so, the parallel computation cost will be $O(n\log n)$.
    \item Scale it to make it complete matrix, $V_{\tau}^* = U_p V_{\tau} t(U_p)$. $U_p$ is $p ^2 \times (\frac{p(p-1)}{2})$ matrix. 
    \item Calculate variance-covariance of $\gamma(K)$, known as $V_{\gamma} = L_p A^{-1} V_{\tau}^* A^{-1} L_p^T$, here, $A$ is Jacobian matrix of order $p^2 \times p^2$, and $L_p$ is scaling matrix of dimension $\frac{p(p-1)}{2} \times p^2$. Inverting $A$ requires inverting $p^2 \times p^2$ diagonal matrix. 
    \item Calculate $\mathbb{D}_g = U_p \Delta_g U_p^T$, where $\Delta_g$ is $\frac{p(p-1)}{2} \times \frac{p(p-1)}{2}$ matrix of $g'(K)$.
    \item Calculate $\Gamma = \frac{dvecl(K)}{dvecl(\gamma(K))}$, which is $\frac{p(p-1)}{2} \times \frac{p(p-1)}{2}$ matrix. 
    \item Calculate $V_\Sigma = \mathbb{D}_g \Gamma$. 
    \item Calculate $\mathbb{D}_{\beta}= ((\Sigma_{22}^{-1}, - (\Sigma_{21} \otimes I_p)(\Sigma_{22}^{-1} \otimes \Sigma_{22}^{-1})) \Tilde{E}$, $\Tilde{E}$ transforms $vec(\Sigma)$ to $(vec(\Sigma_{21}),vec(\Sigma_{22}))$, is a $p(p-1) \times p^2$ matrix. 
    \item Get $\mathbb{D}_\beta U_p \mathbb{D}_g{\Gamma}$ to get $V_\beta$
\end{enumerate}

Now, suppose, we have multiple outcomes ($q \times 1$) and then the asymptotic covariance matrix of the $pq$-length multiple linear regression coefficient of $vec(\hat{\beta}_n) - vec(\beta))$  as - $V_\beta = 
(\mathbb{D}_\beta^{(2)} H_p \mathbb{D}_g{\Gamma} )^T V_{\gamma}\mathbb{D}_\beta^{(2)}H_p \mathbb{D}_g{\Gamma}$, where, $\mathbb{D}_\beta^{(2)} = ((\Sigma_{YY}^{-1} \otimes I_q), - (\Sigma_{YX} \otimes I_p)(\Sigma_{YY}^{-1} \otimes \Sigma_{YY}^{-1})) \Tilde{E}$, where $\Tilde{E}$ transforms $vec(\Sigma)$ to $(vec(\Sigma_{YX}),vec(\Sigma_{XX}))$. Now, suppose we want to test if a subset of the coefficients are $0$, i.e., for a scaling $m \times pq$ matrix $A$, we need to test $A vec(\beta) = 0$, our test statistic would be - $n A vec(\hat{\beta}_n)' (A'V_\beta A)^{-1} A vec(\hat{\beta}_n) \sim  \chi^2_{m}$ under null hypothesis.


\bibliographystyle{agsm}
\bibliography{ref}

\newpage
\begin{center}
{\Large\textbf{Supplementary Materials}}
\end{center}

\beginsupplement

\section{Additional simulation results}
\label{sec:supp-sim}
\begin{figure}[H]
\centering
\begin{subfigure}[h]{  0.95\textwidth}
\centering
\includegraphics[scale=0.3]{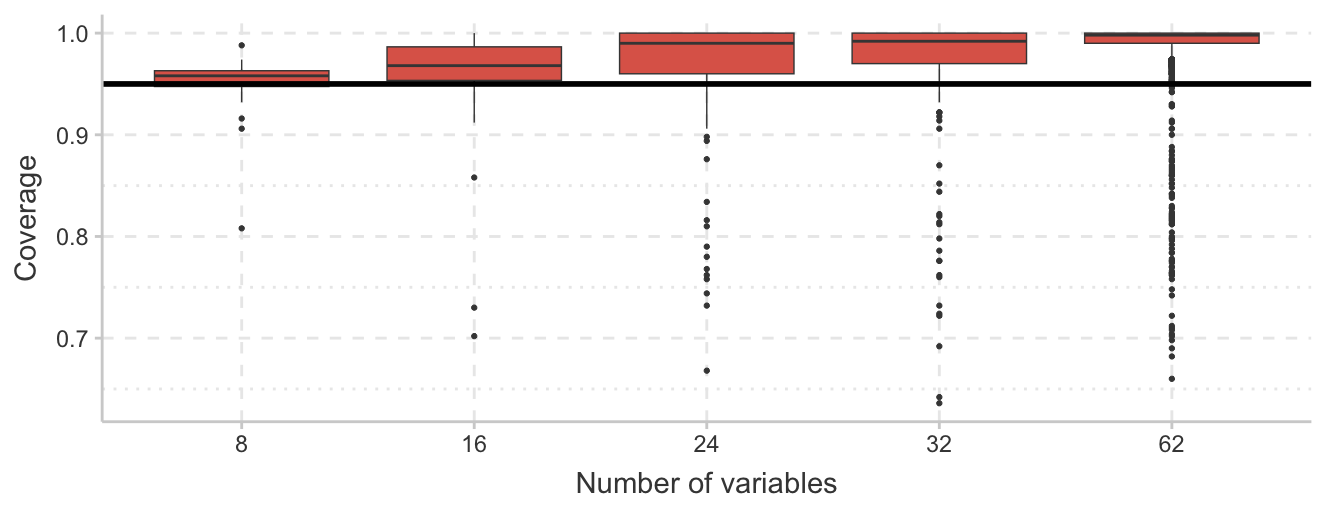}
\caption{Latent correlation}
\label{fig:sim-r-cov}
\end{subfigure}
\centering
\begin{subfigure}[h]{  0.95\textwidth}
\centering
\includegraphics[scale=0.3]{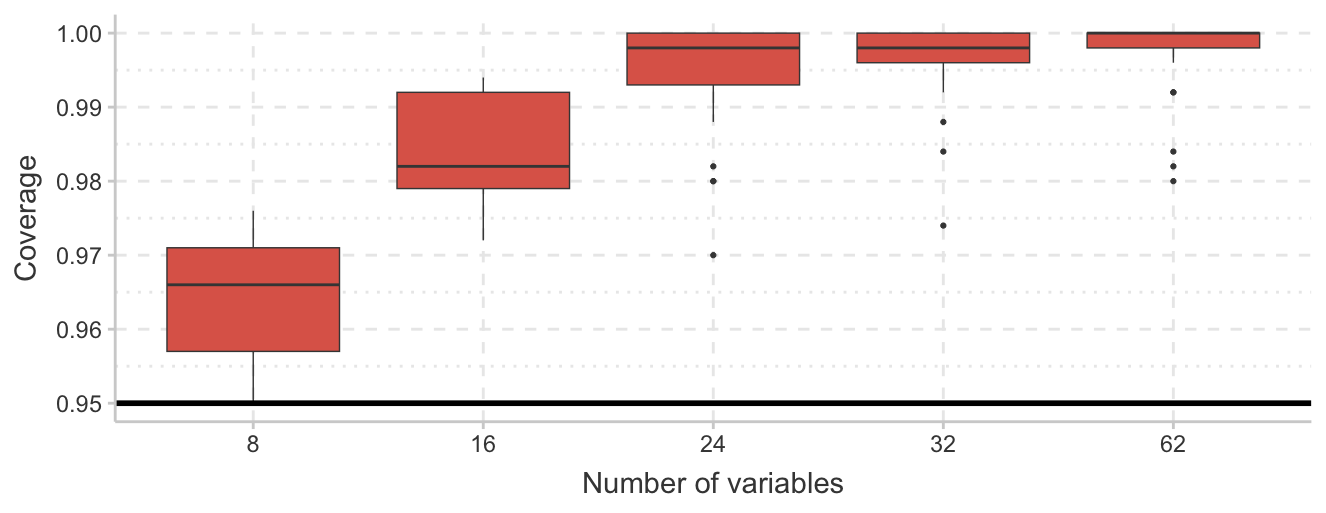}
\caption{The latent regression coefficients}
\label{fig:sim-beta-cov}
\end{subfigure}
     \caption{The boxplot of coverage of estimated parameters (averaged across 500 replicates) for the asymptotic confidence interval. The bold black line denotes $95\% (0.95)$.}
         \label{fig:sim-cov}
\end{figure}
\section{Additional NHANES figures and tables}
\label{sec:supp-nhanes}
\begin{table}[H]
\centering
\resizebox{!}{.30\paperheight}{\begin{tabular}{rlll}
  \hline
 & Variable name & Variable type & NHANES code \\ 
  \hline
1 & Segmented neutrophils percent (\%) & Continuous & LBXNEPCT \\ 
  2 & Alkaline phosphotase (U/L) & Continuous & LBXSAPSI \\ 
  3 & Protein, total (g/L) & Continuous & LBDSTPSI \\ 
  4 & Uric acid (mg/dL) & Continuous & LBXSUA \\ 
  5 & Pulse (heart rate) & Continuous & BPXPLS \\ 
  6 & Mean cell volume (fL) & Continuous & LBXMCVSI \\ 
  7 & Glycohemoglobin & Continuous & LBXGH \\ 
  8 & Iron(ug/dL) & Continuous & LBXSIR \\ 
  9 & Sodium (mmol/L) & Continuous & LBXSNASI \\ 
  10 & Albumin (g/L) & Continuous & LBDSALSI \\ 
  11 & Creatinine (umol/L) & Continuous & LBDSCRSI \\ 
  12 & Cholesterol, total (mmol/L) & Continuous & LBDSCHSI \\ 
  13 & Bilirubin, total (mg/dL) & Continuous & LBXSTB \\ 
  14 & Glucose Serum (mg/dL) & Continuous & LBXSGL \\ 
  15 & Bicarbonate (mmol/L) & Continuous & LBXSC3SI \\ 
  16 & Platelet count (\%) SI & Continuous & LBXPLTSI \\ 
  17 & Red cell distribution width (\%) & Continuous & LBXRDW \\ 
  18 & Lactate Dehydrogenase & Continuous & LBXSLDSI \\ 
  19 & C-reactive protein (mg/dL) & Continuous & LBXCRP \\ 
  20 & Folate, RBC (ng/mL RBC) & Continuous & LBXRBF \\ 
  21 & Blood urea nitrogen (mg/dL) & Continuous & LBXSBU \\ 
  22 & Phosphorus (mmol/L) & Continuous & LBDSPHSI \\ 
  23 & Vitamin D & Continuous & LBDVIDMS \\ 
  24 & Hemoglobin (g/dL) & Continuous & LBXHGB \\ 
  25 & Calcium, total (mg/dL) & Continuous & LBXSCA \\ 
  26 & Folate Serum (nmol/L) & Continuous & LBDFOLSI \\ 
  27 & Red blood cell count (million cells/uL) & Continuous & LBXRBCSI \\ 
  28 & Lymphocyte percent (\%) & Continuous & LBXLYPCT \\ 
  29 & Systolic BP & Continuous & BPXSY \\ 
  30 & Stroke & Binary & MCQ160F \\ 
  31 & Thyroid & Binary & MCQ160M \\ 
  32 & Cancer & Binary & MCQ220 \\ 
  33 & Heart Attack & Binary & MCQ160E \\ 
  34 & Coronary Heart Disease & Binary & MCQ160C \\ 
  35 & Angina pectoris & Binary & MCQ160D \\ 
  36 & Osteoporsis & Binary & OSQ060 \\ 
  37 & Diabetes & Binary & DIQ010 \\ 
  38 & Arthritis & Binary & MCQ160A \\ 
  39 & High blood pressure & Binary & BPQ020 \\ 
  40 & Cough regularly & Binary & RDQ031 \\ 
  41 & Broken or fractured a hip & Binary & OSQ010A \\ 
  42 & Medications & Ordinal & RXDCOUNT \\ 
  43 & Confusion or inability to remember & Binary & PFQ057 \\ 
  44 & Difficulty managing money & Ordinal & PFQ061A \\ 
  45 & Difficulty standing up from armless chair & Ordinal & PFQ061I \\ 
  46 & Difficulty getting in and out of bed & Ordinal & PFQ061J \\ 
  47 & Difficulty standing for about 2 hours & Ordinal & PFQ061M \\ 
  48 & Difficulty stooping, crouching and kneeling & Ordinal & PFQ061D \\ 
  49 & Difficulty grasping or holding small objects & Ordinal & PFQ061P \\ 
  50 & Difficulty lifting/carrying something weighing 10 lbs & Ordinal & PFQ061E \\ 
  51 & Difficulty preparing meals & Ordinal & PFQ061G \\ 
  52 & Difficulty using fork and knife & Ordinal & PFQ061K \\ 
  53 & Difficulty dressing yourself & Ordinal & PFQ061L \\ 
  54 & Difficulty attending social events & Ordinal & PFQ061R \\ 
  55 & Difficulty pushing or pulling large objects & Ordinal & PFQ061T \\ 
  56 & General hearing & Ordinal & AUQ130 \\ 
  57 & Weak or failing kidneys & Binary & KIQ022 \\ 
  58 & Leaked/lost control of urine & Binary & KIQ044 \\ 
  59 & Self-reported health & Ordinal & HUQ010 \\ 
  60 & Health compared to a year ago & Ordinal & HUQ020 \\ 
  61 & Frequency of healthcare use in last year & Ordinal & HUQ050 \\ 
   \hline
\end{tabular}}
\caption{\label{tab:variable-list}The frailty variables we include in our model from NHANES 1999-2010}
\end{table}

\begin{figure}[H]
\centering
\includegraphics[scale=0.3]{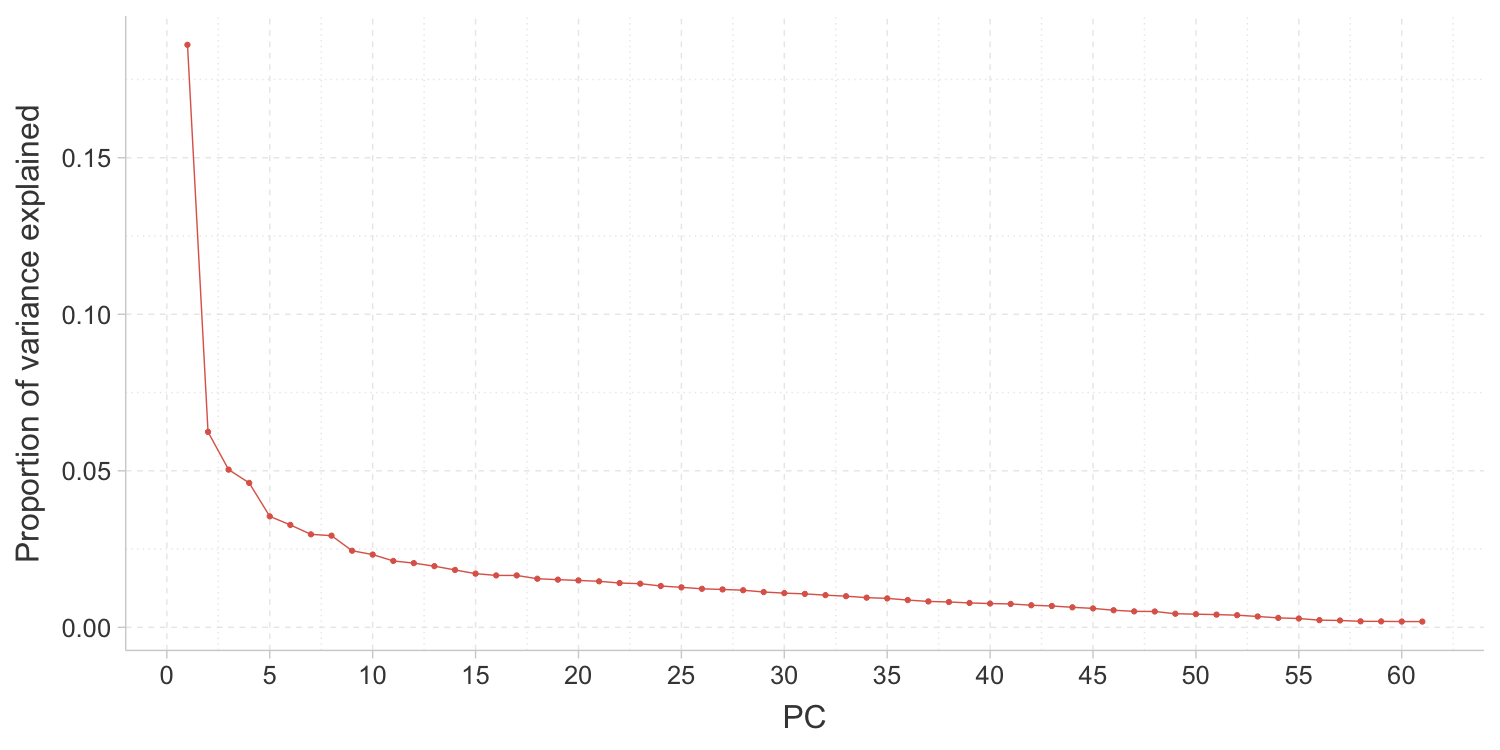}
\caption{The scree plot showing proportion of variance explained for the $61$ PCs}
\label{fig:scree-nhanes}
\end{figure}

\begin{figure}[H]
     \centering
     \begin{subfigure}[b]{  0.97\textwidth}
         \centering
         \includegraphics[width=\textwidth]{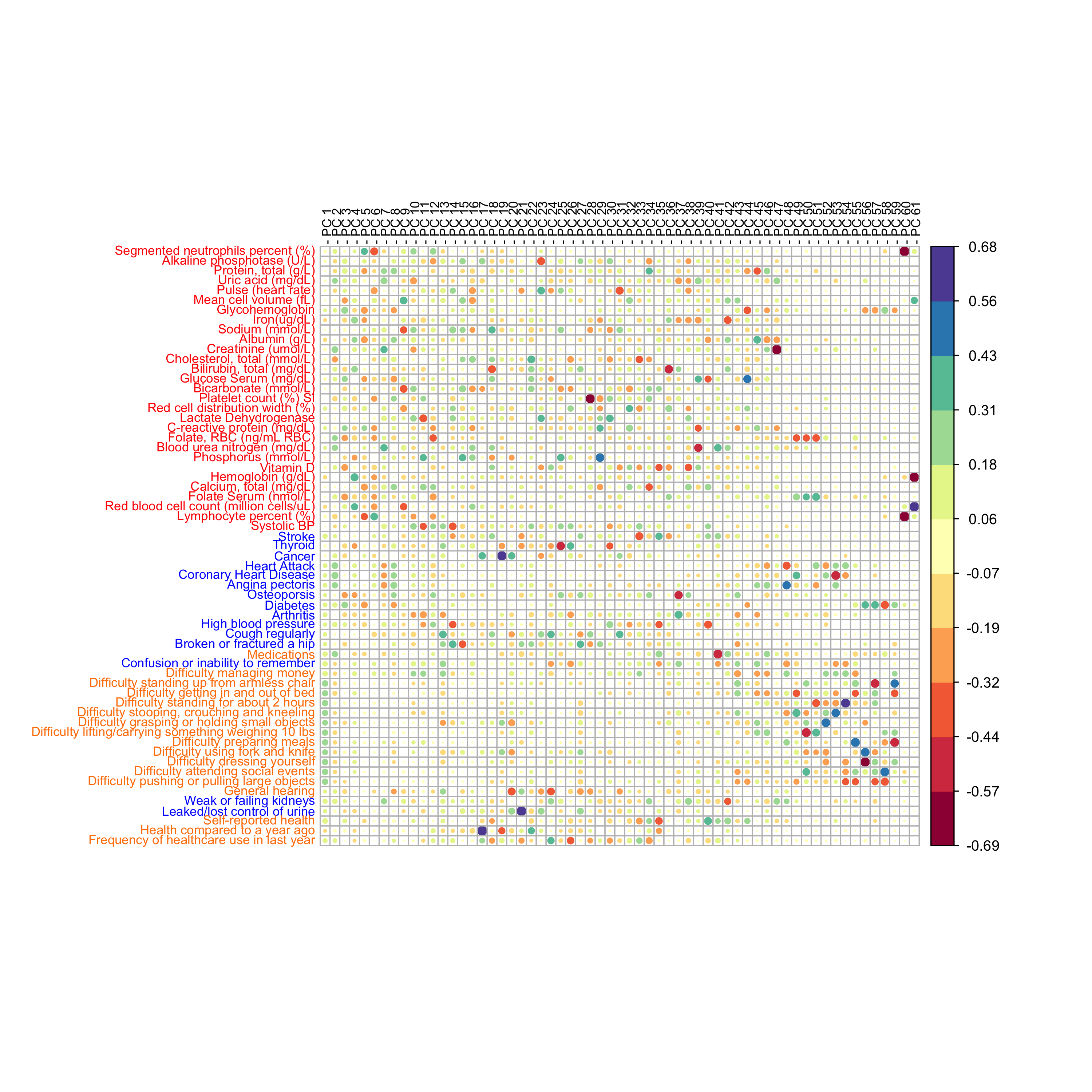}
         \caption{The heatamp showing PC loadings for latent correlation matrix corresponding to different variables}
         \label{fig:pcloading-nhanes}
     \end{subfigure}
     \hfill
     \begin{subfigure}[b]{0.65\textwidth}
         \centering
         \includegraphics[width=\textwidth]{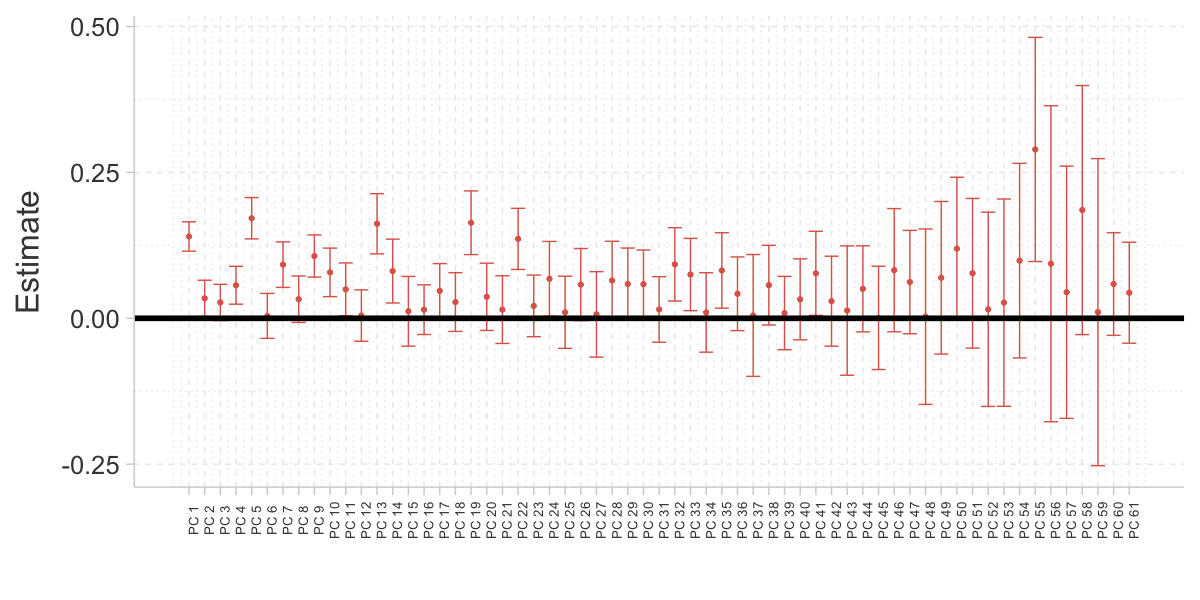}
         \caption{The latent principal component Regression coefficients for outcome mortality}
         \label{fig:latpcreg-nhanes}
     \end{subfigure}
     \caption{SGC-PCA and SGC-PCR results for the NHANES frailty application.}
    \label{fig:nhanes-pca-pcr}
\end{figure}

\end{document}